\documentclass[12pt  en  cite=authoryear]{elegantpaper}

\title{Non-finite Axiomatizability and Undecidability of $\mathsf{Cheq}$}
\author{Han Xiao 
 \thanks{None of the mathematical arguments, proofs presented in this paper were generated or developed using any AI technology. All such arguments, proofs were independently conceived and produced by the author. It should be emphasized that the non-finite axiomatizability result was established before the rise of AI, and its proof is self-contained. Regarding Theorem~\protect\ref{cheqund}, the undecidability, however, does use the AI-generated proof that Medvedev logic is undecidable, see \protect\cite{Ak26, Paw26}.
}}
\institute{\footnotesize The Tsinghua-UvA JRC for Logic, Tsinghua University\\
\texttt{sheropen@163.com}}

\date{}

\usepackage{array}
\usepackage{tikz}
\usepackage{pgfplots}
\usepackage{etoolbox}
\usetikzlibrary{calc}
\usepackage{xcolor}

\renewenvironment{thebibliography}[1]{%
  \begin{oldthebibliography}{#1}%
  \small
}%
{%
  \end{oldthebibliography}%
}

\usepackage{iftex}

\usepackage{tikz-cd}
\usepackage{tikz}
\usepackage{amsmath,amssymb,xcolor,amsthm}
\usetikzlibrary{decorations.pathmorphing, decorations.markings}
\usepackage[all]{xy}
\usepackage{ifthen,xspace}
\usepackage{graphicx,color}
\usepackage{MnSymbol}
\usepackage{float}
\usepackage{hyperref}

\ifpdf
  \usepackage{underscore}         
  \usepackage[T1]{fontenc}        
\else
  \usepackage{breakurl}           
\fi

\begin{document}

\maketitle
\begin{abstract}
We prove that $\mathsf{Cheq}$ is not finitely axiomatizable, resolving a longstanding open problem in intermediate and modal logics. We further prove that the undecidability of Medvedev logic implies the undecidability of $\mathsf{Cheq}$.
\end{abstract}

\section{Introduction}
\label{sec:intro}
In 2003, the \emph{modal logic of chequered subsets} $\mathsf{L}_{\infty}$ was introduced by van Benthem, Guram Bezhanishvili, and Gehrke \cite[in particular, p.~343]{vBBG03}.   
Whether $\mathsf{L}_{\infty}$ is finitely axiomatizable is a longstanding open problem. 
This logic plays an important role in the study of modal logics of 
various topological spaces; see, for example, \cite{vBB07, AvBB03}.

Litak \cite{Lit04} later denoted its 
intermediate analogue by $\mathsf{Cheq}$. He proved that $\mathsf{Cheq}$ has disjunction property and validates the Scott axiom, while it does not validate the Kreisel-Putnam axiom and is not structurally complete. Most importantly, $\mathsf{Cheq}$ is contained in the Medvedev logic $\mathsf{Med}$, another well-known intermediate logic. In 2006, Fontaine improved the result by Maksimova, Skvortsov, and Shehtman and proved that the Medvedev logic is not even finitely axiomatizable over $\mathsf{Cheq}$. Analogues of this results for the logic of spiked Boolean algebras and the generalized Medvedev logics are given in \cite{LX24} and \cite{Xiao26}, respectively. In 2007, Shatrov claimed to prove the non-finitely axiomatizability of $\mathsf{Cheq}$ at the conference Algebraic and Topological Methods in Non-Classical Logics (TANCL 07) in Oxford, but it has never been published. Using Offner's result \cite{Off08}, Kuznetsov in \cite{Kuz19} examined Fontaine's and Shatrov's strategies and proposed an alternative solution as an open problem for proving the non-finitely axiomatizability of $\mathsf{Cheq}$. 
 
Beyond its role in the non-classical logics, $\mathsf{L}_{\infty}$ also arises in the study of \emph{modal logics of forcing}. This line of research, initiated by Hamkins and L{\"o}we \cite{HL08}, investigates the modal logic of the class of models of set theory with the forcing relation. In particular, determining the \emph{modal logics of c.c.c.\ forcing} has become a central problem in this area. Building on ideas of Inamdar, we proved in \cite{Xiao24} that, under an additional set theoretical assumption, the modal logic of c.c.c.\ forcing is contained in $\mathsf{S4.FPFA}$, the least modal companion of $\mathsf{Cheq}$.

\subsection{Outline of the paper}
The Section~\ref{sec:def} introduces the definitions and tools for the results of this paper. We introduce Medvedev logic $\mathsf{Med}$, Inamdar logic $\mathsf{Inam}$, and generalized Medvedev logic in Section~\ref{sec:relatedlogics}, and give a comparison of the modal companions of those logics in Theorem~\ref{thm:fpfatba} and Theorem~\ref{thm:fpfasba}. We then discuss a geometric perspective for $\mathsf{Cheq}$ in Section~\ref{sec:geo}. Following Kuznetsov's idea of applying Offner's edge-coloring result \cite{Off08} to $\mathsf{Cheq}$ from his unpublished note, we discuss the two strategies of Fontaine and Shatrov, as well as Kuznetsov's alternative solution. Finally, we give the self contained proof for the non-finite axiomatizability of $\mathsf{Cheq}$ in Theorem~\ref{thm:cheq}. We also prove that the undecidability of Medvedev logic implies the undecidability of $\mathsf{Cheq}$ in the Corollary of Theorem~\ref{thm:cheqmed}.

\section{Definitions}
\label{sec:def}

Throughout this paper, we assume familiarity with the basic theory of intermediate and modal logics. 
We write $\mathsf{IPC}$ for \emph{intuitionistic propositional calculus}. 
By an (intuitionistic Kripke) frame we mean a poset. We use the terms `frame' and `poset' interchangeably and all posets considered are finite.

Given a frame $\mathbf{F}$, we write $\mathsf{Log}(\mathbf{F})$ for the propositional logic consisting of all formulas valid on the frame $\mathbf{F}$; if
$\mathsf{L} \subseteq \mathsf{Log}(\mathbf{F})$,
we call $\mathbf{F}$ an \emph{$\mathsf{L}$-frame}. If $\mathcal{C}$ is a class of frames, we write $\mathsf{Log}(\mathcal{C})$ for the intermediate logic consisting of all formulas valid on every frame of $\mathcal{C}$; in this case, we say that $\mathcal{C}$ \emph{characterizes} the logic.
\subsection{Posets as Kripke frames} 

Let $\mathbf{F} = \langle F,\leq \rangle$ be a finite poset; we write $x<y$ if $x\leq y$ and $x\neq y$. An element $x \in F$ is called \emph{maximal} if there is no $y \in F$ such that $x < y$; it is called a \emph{top} (or \emph{greatest element}) if everything is below it. A top is necessarily unique.
We say that $\mathbf{F}$ is a frame \emph{with a top} if it contains a top.
If $x<y$, we call $y$ an \emph{immediate successor} of $x$ if there is no $z$ such that $x<z<y$. The \emph{branching degree} of $x$, denoted by $b(x)$, is defined as the number of immediate \mbox{successors of $x$}. 

The notation $x{\uparrow}$ (resp.\ $x{\downarrow}$) denotes the \emph{upset} $\{y\in F : x\leq y\}$ (resp.\ the \emph{downset} $\{y\in F : x\geq y\}$). A subframe $\mathbf{S}=\langle S, \leq \rangle$ of $\mathbf{F} = \langle F,\leq \rangle$ is called a \emph{generated subframe} if $S$ is upward closed in $F$; we say that $\mathbf{S}$ is \emph{generated by $x$} if $S = x{\uparrow}$. If $\mathbf{F}$ is generated by $x$, then $\mathbf{F}$ is called \emph{rooted} and $x$ the \emph{root} of $\mathbf{F}$. 
The \emph{depth} of a frame $\mathbf{F}$ is the length of the longest 
chain in $\mathbf{F}$. 
For $x \in F$, the \emph{depth of $x$}, denoted by $d(x)$, is the depth of the subframe 
generated by $x$.

Let $\mathbf{F}=\langle F,\le_0\rangle$ and 
$\mathbf{G}=\langle G,\le_1\rangle$ be posets. 
Their \emph{product} is the poset 
$\mathbf{F}\times\mathbf{G}
:= \langle F\times G,\le\rangle$, 
where the order $\le$ is defined coordinatewise by $(x,y)\le (x',y')$ if and only if $x\le_0 x'$ and $y\le_1 y'$.
For a poset $\mathbf{F}$, we write 
$\mathbf{F}^2 := \mathbf{F}\times\mathbf{F}$ and define inductively
$\mathbf{F}^n := \mathbf{F}^{\,n-1}\times\mathbf{F}$. Now let $\mathbf{F}$ be a finite frame with a top. Its \emph{topless frame} $\mathbf{F}^{-}$ is the frame obtained from $\mathbf{F}$ by removing its top element.
For $n\ge1$, the \emph{$n$-fold topless product} of $\mathbf{F}$ is $\mathbf{F}^{n-} := (\mathbf{F}^n)^{-}$, that is, the topless frame of the product $\mathbf{F}^n$.

We write $\mathbf{F}\oplus\mathbf{G}$ for their \emph{linear sum}, 
obtained by placing all elements of $\mathbf{G}$ strictly above all 
elements of $\mathbf{F}$. 
We define the iterated finite sum by recursion via $\mathbf{F}\times 1 := \mathbf{F}$ and $\mathbf{F}\times (n+1) := (\mathbf{F}\times n)\oplus \mathbf{F}$. 
The \emph{vertical sum} of $\mathbf{F}$ and $\mathbf{G}$ is obtained from the linear sum of $\mathbf{F}$ and $\mathbf{G}$ by identifying the top of $\mathbf{F}$ with the root of $\mathbf{G}$ (provided they exist).

\subsection{Morphisms} 
\begin{definition}
If $\mathbf{F}$ and $\mathbf{G}$ are frames and $f:F\to G$ is a map, we call $f$ a \emph{p-morphism} if
\begin{enumerate}
	\item  for all $x,y\in F$, if $x\leq y$, then $f(x)\leq f(y)$,
	\item  if $f(x)\leq u$, then there is some $y\geq x$ such that $f(y) = u$.
\end{enumerate}	

If $f$ is onto, then we call $\mathbf{G}$ a \emph{p-morphic image} of $\mathbf{F}$. 
\end{definition}
\begin{definition}
We call $f$ is a \emph{p-morphism} from a model $(\mathbf{F}, V_0)$ on a model $(\mathbf{G},V_1)$ if $f$ is a \emph{p}-morphism from $\mathbf{F}$ to $\mathbf{G}$ such that $x\in V_0(p)$ iff $f(x)\in V_1(p)$ for all $x\in \mathbf{F}$ and all variables $p$.	
\end{definition} 
 
\begin{theorem}[Folklore]
If $\mathbf{G}$ is a $p$-morphic image of $\mathbf{F}$, then $\mathbf{F}\vDash \varphi \text{ implies } \mathbf{G}\vDash \varphi$, for any formula $\varphi$.\\
If $f$ is a $p$-morphism from a model $\mathcal{M}_0$ to a model $\mathcal{M}_1$, then $x\vDash_{\mathcal{M}_0} \varphi$ iff $f(x)\vDash_{\mathcal{M}_1} \varphi$, for any $x$ and formula $\varphi$.
\end{theorem}


\subsection{Modal companions}
In the 1930s, G\"odel suggested a translation $\mathrm{T}$ which can embed $\mathsf{IPC}$ into the modal logic $\mathsf{S4}$ in \cite{God33}.

\begin{definition}
The \emph{G\"odel translation} $\mathrm{T}$ is the map defined by
\[
\begin{aligned}
\mathrm{T}(\bot)&=\bot,
&\qquad
\mathrm{T}(p)&=\Box p \quad (p\in\mathbf{Prop}),\\
\mathrm{T}(\varphi_0\wedge\varphi_1)
&=\mathrm{T}(\varphi_0)\wedge\mathrm{T}(\varphi_1),
&
\mathrm{T}(\varphi_0\vee\varphi_1)
&=\mathrm{T}(\varphi_0)\vee\mathrm{T}(\varphi_1),\\
\mathrm{T}(\varphi_0\to\varphi_1)
&=\Box\bigl(\mathrm{T}(\varphi_0)\to\mathrm{T}(\varphi_1)\bigr).
\end{aligned}
\]
\end{definition}
In the 1940s, McKinsey and Tarski proved that the G\"odel translation embeds $\mathsf{IPC}$ to $\mathsf{S4}$ in \cite[\S~5]{MT48}. In the 1950s, Dummett and Lemmon extend this result to intermediate logics and extensions of $\mathsf{S4}$ in \cite[Theorem 1]{DL59}. In the 1970s, Esakia \cite{Esa79-1,Esa79-2} developed the theory of Heyting algebras; this was independently done by Maksimova and Rybakov \cite{MR74} and also by Blok \cite{Blo76}.
\begin{definition}
A modal logic $\mathsf{M}\supseteq\mathsf{S4}$ is called a \emph{modal companion} of an intermediate logic $\mathsf{L}$, if the following holds: $\varphi \in \mathsf{L}$ iff  $\mathrm{T}(\varphi) \in \mathsf{M}$, for each intuitionistic formula $\varphi$. 	
\end{definition}

\begin{theorem}[Folklore; cf.\ {\cite[Proposition 7]{She90}}]
 For any intermediate logic \( \mathsf{L} \), both the least and the greatest modal companions exist. Let $\tau(\mathsf{L})$ be the least modal companion and 
\[\tau(\mathsf{L})=\mathsf{S4} + \{\mathrm{T}(\varphi) : \varphi \in \mathsf{L}\},
\] 
let $\sigma(\mathsf{L})$ be the greatest modal companion and 
\[\sigma(\mathsf{L}) =\mathsf{Grz} + \{\mathrm{T}(\varphi) : \varphi \in \mathsf{L}\}.
\]	
\end{theorem}

\begin{theorem}
	\begin{enumerate}
		\item The map $\tau$ is an isomorphism of the lattice of intermediate logics into the lattice of \mbox{extensions of $\mathsf{S4}$}.
		\item(The Blok-Esakia theorem) The map $\sigma$ is an isomorphism from the lattice of intermediate logics onto the lattice of extensions of $\mathsf{Grz}$.
	\end{enumerate}
\end{theorem}

If $\mathbf{P}$ is a finite poset, we say that $\mathbf{P}^{\bullet}$ is a \emph{thickening of $\mathbf{P}$} if $\mathbf{P}^{\bullet}$ is a finite partial pre-order, $\sim$ is the induced equivalence relation  and $\mathbf{P} = \mathbf{P}^{\bullet}/{\sim}$. Given a class $\mathcal{C}$ of finite posets, let $\mathcal{C}^{\bullet}$ be the class of thickenings of elements of it. We write $\mathsf{ML}(\mathcal{C})$ and $\mathsf{ML}(\mathcal{C}^\bullet)$ for the modal logic consisting of all modal formulas valid on every element of $\mathcal{C}$ and $\mathcal{C}^\bullet$, respectively.

\begin{theorem}[Esakia; cf.\ {\cite[Proposition 9]{She90}}]
	For any class $\mathcal{C}$ of finite posets, the modal logic of $\mathcal{C}$ is the greatest modal companion of its logic, that is, $\mathsf{ML}(\mathcal{C})=\sigma(\mathsf{Log}(\mathcal{C}))$. 
\end{theorem}
\begin{theorem}[Zakharyaschev; cf.\ {\cite[Proposition 10]{She90}}]
	For any class $\mathcal{C}$ of finite posets, the modal logic of $\mathcal{C}^{\bullet}$ is the least modal companion of its logic, that is, $\mathsf{ML}(\mathcal{C}^{\bullet})=\tau(\mathsf{Log}(\mathcal{C}))$. 
\end{theorem}
The modal logic $\mathsf{L}_{\infty}$ is the greatest modal companion of $\mathsf{Cheq}$, that is $\mathsf{L}_{\infty}=\sigma(\mathsf{Cheq})$. 
\begin{definition}
The modal logic $\mathsf{S4.FPFA}$ is the least modal companion of $\mathsf{Cheq}$, that is $\mathsf{S4.FPFA}=\tau(\mathsf{Cheq})$.
\end{definition}
\begin{lemma}[Folklore]
	Let $\mathsf{L}$ be an intermediate logic, then $\mathsf{L}$ is finitely axiomatizable iff $\sigma(\mathsf{L})$ is finitely axiomatizable.
\end{lemma}
This is \cite[Corollary 8]{She90}, where the result is attributed to Maksimova. Thus, if we already have $\tau(\mathsf{L})$ is finitely axiomatizable, then $\sigma (\mathsf{L}) = \mathsf{Grz} + \tau(\mathsf{L})$ is and so $\mathsf{L}$ is finitely axiomatizable. As a consequence, it is enough to show that the intermediate logic $\mathsf{L}$ is not finitely axiomatizable in order to have all three logics $\mathsf{L}$, $\sigma(\mathsf{L})$, and $\tau(\mathsf{L})$ are not finitely axiomatizable.

\subsection{Friedman translation}
\begin{definition}
Let $\rho$ be any formula. For a formula $\varphi$, the \emph{Friedman translation} of $\varphi$ by $\rho$, $\varphi^{\rho}$, is defined inductively as follows:
\begin{itemize}
	\item $\varphi^\rho: = \varphi \vee \rho$, if $\varphi$ is atomic, 
	\item $(\varphi_0 \circ \varphi_1)^{\rho}:= \varphi_0^{\rho} \circ \varphi_1^{\rho}$, for $\circ\in \{\wedge,\vee, \to\}$.
\end{itemize}
\end{definition}

We now introduce the First Pruning Lemma in \cite{vDMKV86}. 
\begin{lemma}
	For a give Kripke model $\mathcal{M}$, let $\mathcal{M}^{\rho}$ be the submodel obtained from $\mathcal{M}$ by removing all points that force $\rho$. Then for any $x\in \mathcal{M}^{\rho}$ and any formula $\varphi$, $x \vDash_{\mathcal{M}}  \varphi^{\rho}$ iff $x \vDash_{\mathcal{M}^{\rho}} \varphi$.
\label{lem:pruning} 
\end{lemma}

\subsection{The logic $\mathsf{Cheq}$}

Let $\mathbf{V}$ be the \emph{two-fork}, that is, the poset consisting 
of three elements $\{w_0,w_1,w_2\}$ such that $w_0$ is the root and $w_1$ and $w_2$ 
are incomparable maximal elements. Define recursively 
$\mathbf{V}_1 := \mathbf{V}$ and 
$\mathbf{V}_{n+1} := \mathbf{V}_n \times \mathbf{V}$,
where $\times$ denotes the usual product of posets. 
Every point $x$ of $\mathbf{V}_n$ can be associated with an $n$-tuple $(x_0, x_1, \ldots, x_{n-1})$, where $x_i \in \{w_0, w_1, w_2\}$. 
The \emph{logic $\mathsf{Cheq}$} is the intermediate logic characterized by 
the class $\{\mathbf{V}_n : n \in \omega\}$.

\begin{definition}
A poset $\mathbf{A}_n$ is called the \emph{finite partial function algebra} on $n$ elements if it consists of partial functions from $[n]$ to $\{1, 2\}$. For $a, b \in \mathbf{A}_n$, we define $a \leq b \in \mathbf{A}_n$ if and only if $a = b{\upharpoonright}\mathrm{dom}(a)$. For any point $a\in \mathbf{A}_n$, we associate the pair $(1_a, 2_a)$ where $1_a=a^{-1}(1)$ and $2_a=a^{-1}(2)$. 
\end{definition}

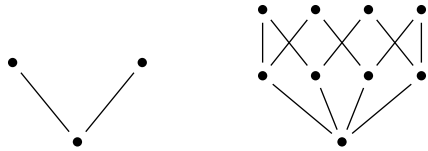
\begin{figure}[htbp]
\centering
\begin{tikzpicture}[
scale=0.35,
every node/.style={
    circle,
    fill=black,
    minimum size=3.5pt,
    inner sep=0pt
},
every label/.style={
    circle=none,
    fill=none,
    inner sep=1pt
},
edge/.style={
    line width=0.5pt,
    shorten >=3pt,
    shorten <=3pt
}
]

\begin{scope}[xshift=0cm]
    \node (a0)  at (0,1) {};
    \node (a1)  at (-2.45,4) {};
    \node (a2)  at (2.45,4) {};

    \draw[edge] (a0) -- (a1);
    \draw[edge] (a0) -- (a2);
\end{scope}

\begin{scope}[xshift=10cm]
    \node (b00)  at (0,1) {};

    \node (b01)  at (-3,3.5) {};
    \node (b10)  at (-1,3.5) {};
    \node (b20)  at (1,3.5) {};
    \node (b02)  at (3,3.5) {};

    \node (b11)  at (-3,6) {};
    \node (b21)  at (-1,6) {};
    \node (b12)  at (1,6) {};
    \node (b22)  at (3,6) {};

    \draw[edge] (b00) -- (b01);
    \draw[edge] (b00) -- (b10);
    \draw[edge] (b00) -- (b20);
    \draw[edge] (b00) -- (b02);

    \draw[edge] (b01) -- (b11);
    \draw[edge] (b01) -- (b21);

    \draw[edge] (b10) -- (b11);
    \draw[edge] (b10) -- (b12);

    \draw[edge] (b20) -- (b21);
    \draw[edge] (b20) -- (b22);

    \draw[edge] (b02) -- (b12);
    \draw[edge] (b02) -- (b22);
\end{scope}

\end{tikzpicture}

\caption{$\mathbf{V}_1$ and $\mathbf{V}_2$.}
\label{fig:C1C2}
\end{figure}

\begin{lemma}
$\mathbf{A}_n \cong \mathbf{V}_n$	
\label{AcongV}
\end{lemma}

\begin{proof}
Definite a map $f$ from $\mathbf{A}_n$ to $\mathbf{V}_n$ as follows: 

\[f(a)=x \text{ and } x_i=\begin{cases}
w_0, &  i \notin  \mathrm{dom}(a)\\
w_1, &  i \in 1_a\\
w_2, &   i \in 2_a\\
\end{cases}, \text{ for } 1\leq i \leq n.\]

Assume that $f(a)=x=(x_1,x_2,\ldots,x_n)$, $f(b)=y=(y_1,y_2,\ldots,y_n)$. Then $a \leq b$ in $\mathbf{A}_n$  iff  $a=b{\upharpoonright} \mathrm{dom}(a)$ iff $x_i\leq y_i$ iff $f(a)\leq f(b)$ in $\mathbf{V}_n$. Because $f$ is one-to-one from $\mathbf{A}_n$ to $\mathbf{V}_n$, then it is an isomorphism from $\mathbf{A}_n$ to $\mathbf{V}_n$, thus $\mathbf{A}_n\cong\mathbf{V}_n$ for any $n\in \omega$.
\end{proof}

\section{Related logics}
\label{sec:relatedlogics}
\subsection{Medvedev logic}
We first recall the \emph{Medvedev logic}, which might be one of most famous intermediate logics.
In 1932, Kolmogorov \cite{Kol32} suggested a constructive interpretation of intuitionistic logic as a calculus of problems. To make precise the description proposed by Kolmogorov, Medvedev established the foundational framework for the logic of finite problems in \cite{Med62}. We recall the definition of the \emph{Medvedev frame} and the corresponding \emph{Medvedev logic}.
\begin{definition}
For $n \ge 1$, the \emph{Medvedev frame on $n$ atoms} $\mathbf{P}_0(n)$ is defined by
$\mathbf{P}_0(n) 
:= \bigl\langle 
\{ X \subseteq [n] : X \neq \varnothing \}, 
\supseteq 
\bigr\rangle$.
The \emph{Medvedev logic} $\mathsf{Med}$ is the intermediate logic of all Medvedev frames, that is, $\mathsf{Med}:= \mathsf{Log}\bigl(\{ \mathbf{P}_0(n) : n \in \omega \}\bigr)$.
\end{definition}
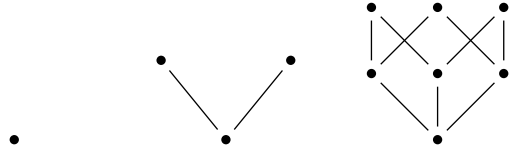
\begin{figure}[htbp]
\centering
\begin{tikzpicture}[
scale=0.35,
every node/.style={
    circle,
    fill=black,
    minimum size=3.5pt,
    inner sep=0pt
},
edge/.style={
    line width=0.5pt,
    shorten >=3pt,
    shorten <=3pt
}
]

\begin{scope}[xshift=0cm] 
    \node (a1) at (0,1) {};
    
\end{scope}

\begin{scope}[xshift=8cm] 
    \node (b1) at (0,1) {};
    \node (b2) at (-2.45,4) {};
    \node (b3) at (2.45,4) {};

    \draw[edge] (b1) -- (b2);
    \draw[edge] (b1) -- (b3);
\end{scope}

\begin{scope}[xshift=16cm] 
    \node (c1) at (0,1) {};
    \node (c2) at (-2.5,3.5) {};
    \node (c3) at (0,3.5) {};
    \node (c4) at (2.5,3.5) {};
    \node (c5) at (-2.5,6) {};
    \node (c6) at (0,6) {};
    \node (c7) at (2.5,6) {};

    \draw[edge] (c1) -- (c2);
    \draw[edge] (c1) -- (c3);
    \draw[edge] (c1) -- (c4);
    \draw[edge] (c2) -- (c5);
    \draw[edge] (c2) -- (c6);
    \draw[edge] (c3) -- (c5);
    \draw[edge] (c3) -- (c7);
    \draw[edge] (c4) -- (c6);
    \draw[edge] (c4) -- (c7);
\end{scope}

\end{tikzpicture}
\caption{Medvedev frames.}
\label{fig:med}
\end{figure}

\begin{lemma}[Maksimova, Skvortsov, \& Shehtman; cf.\ {\cite[Lemma 4]{MSS79}}]
If $\mathbf{F}$ is a finite rooted frame with a top, then $\mathbf{F}$ is a $\mathsf{Med}$-frame.		
\label{lem:topmedframe}
\end{lemma}

Inheriting such a tradition on the frame for logic of finite problems, the order in our \emph{finite Boolean algebra} here is reverse inclusion $\supseteq$. 

\begin{definition}
For $n \ge 1$, the \emph{Boolean algebra on $n$ atoms} $\mathbf{P}(n)$ is defined by $\mathbf{P}(n):=\bigl\langle 
\mathcal{P}([n]), \supseteq \bigr\rangle$, where $\mathcal{P}([n])$ denotes the powerset of $[n]$.
\end{definition}

\subsection{Inamdar logic}
In his Master's thesis \cite{Ina13}, Inamdar introduced a different but closely related class of structures that he called \emph{spiked Boolean algebras} and proved that their modal logic $\mathsf{S4.sBA}$ is an upper bound for the modal logic of c.c.c.\ forcing \cite[\S~5]{Ina13}.

\begin{definition}
If $\mathbf{P}(n)$ is a finite Boolean algebra with $n$ co-atoms $\{\{j\} : 1 \leq j \leq n\}$, we define a
\emph{spiked Boolean algebra} $\mathbf{S}(n)$ by adding $n$ additional nodes $\{\{js\} : 1 \leq j \leq n\}$ such that for
every $b \in \mathbf{P}(n)$, we have $b \leq js$ if and only if $b \leq j$. The \emph{Inamdar logic} is the intermediate logic of all spiked Boolean algebras, that is, $\mathsf{Inam}:= \mathsf{Log}\bigl(\{ \mathbf{S}(n) : n \in \omega \}\bigr)$.  	
\end{definition}
The spiked Boolean algebras for the Boolean algebras
with two, four and eight elements can be seen in Figure \ref{fig:sBa}. 

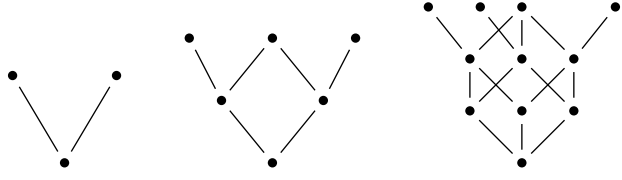
\begin{figure}[htbp]
\centering
\begin{tikzpicture}[
scale=0.275,
every node/.style={
    circle,
    fill=black,
    minimum size=3.5pt,
    inner sep=0pt
},
edge/.style={
    line width=0.5pt,
    shorten >=3pt,
    shorten <=3pt
}
]

\begin{scope}[xshift=0cm] 
    \node (a1) at (0,1) {};
    \node (a2) at (-2.5,5.2) {};
    \node (a3) at (2.5,5.2) {};

    \draw[edge] (a1) -- (a2);
    \draw[edge] (a1) -- (a3);
\end{scope}

\begin{scope}[xshift=10cm] 
    \node (b1) at (0,1) {};
    \node (b2) at (-2.45,4) {};
    \node (b3) at (2.45,4) {};
    \node (b4) at (0,7) {};
    \node (b5) at (-4,7) {};
    \node (b6) at (4,7) {};

    \draw[edge] (b1) -- (b2);
    \draw[edge] (b1) -- (b3);
    \draw[edge] (b2) -- (b4);
    \draw[edge] (b2) -- (b5);
    \draw[edge] (b3) -- (b4);
    \draw[edge] (b3) -- (b6);
\end{scope}

\begin{scope}[xshift=22cm] 
    \node (c1) at (0,1) {};
    \node (c2) at (-2.5,3.5) {};
    \node (c3) at (0,3.5) {};
    \node (c4) at (2.5,3.5) {};
    \node (c5) at (-2.5,6) {};
    \node (c6) at (0,6) {};
    \node (c7) at (2.5,6) {};
    \node (c8) at (0,8.5) {};
    \node (c9) at (-4.5,8.5) {};
    \node (c10) at (-2,8.5) {};
    \node (c11) at (4.5,8.5) {};

    \draw[edge] (c1) -- (c2);
    \draw[edge] (c1) -- (c3);
    \draw[edge] (c1) -- (c4);

    \draw[edge] (c2) -- (c5);
    \draw[edge] (c2) -- (c6);

    \draw[edge] (c3) -- (c5);
    \draw[edge] (c3) -- (c7);

    \draw[edge] (c4) -- (c6);
    \draw[edge] (c4) -- (c7);

    \draw[edge] (c5) -- (c8);
    \draw[edge] (c6) -- (c8);
    \draw[edge] (c7) -- (c8);

    \draw[edge] (c5) -- (c9);
    \draw[edge] (c6) -- (c10);
    \draw[edge] (c7) -- (c11);
\end{scope}

\end{tikzpicture}

\caption{Spiked Boolean algebras.}
\label{fig:sBa}
\end{figure}

\subsection{Generalized Medvedev logics}
A useful structural perspective on Medvedev logic is that it is characterized by \emph{topless products} of the $2$-chain. More precisely, the Medvedev frame $\mathbf{P}_0(n)$ is obtained from the $n$-fold product of the $2$-chain by removing its greatest element. 
From this perspective, it is natural to ask whether these constructions lead to new intermediate logics with similar structural properties. 
In particular, one may replace the $2$-chain by more general finite frames and consider the topless products obtained by removing the greatest element from their finite products. This idea was proposed by Nick Bezhanishvili, who asked whether the classical non-finite axiomatizability result for $\mathsf{Med}$ extends to such generalized constructions. 
More specifically, he conjectured that non-finite axiomatizability should already hold for topless products of arbitrary  non-singleton finite chains and, more generally, for topless products of arbitrary non-singleton finite rooted frames with a top.
These considerations lead naturally to a broader class of intermediate logics that we call \emph{generalized Medvedev logics} in \cite{Xiao26}. 
 
\begin{definition}
Let $\mathbf{F}$ be a non-singleton finite rooted frame with a top. We define $\mathsf{TLP}_{\mathbf{F}}
:=
\mathsf{Log}\bigl(\{ \mathbf{F}^{n-} : n \in \omega \}\bigr).$ The logic $\mathsf{TLP}_{\mathbf{F}}$ is called the 
\emph{generalized Medvedev logic over $\mathbf{F}$}.
\end{definition}

\subsection{Comparison of the relevant logics}
With the classes of finite partial function algebras, finite topless Boolean algebras and finite spiked Boolean algebras, we obtain corresponding intermediate and modal logics as listed in the following table~\ref{tab:comparison}. 
As usual, we call $\mathsf{Log}(\{\mathbf{A}_n: n\geq 1\})$ the logic $\mathsf{Cheq}$, $\mathsf{ML}(\{\mathbf{A}_n: n\geq 1\})$  the modal logic of chequered sets $\mathsf{L}_{\infty}$, and $\mathsf{Log}(\{\mathbf{P}_0(n): n\geq 1\})$ Medvedev logic $\mathsf{Med}$. 
We shall call $\mathsf{Inam}:=\mathsf{Log}(\{\mathbf{S}(n): n\geq 1\})$ Inamdar logic in recognition of the fact that Inamdar introduced the finite spiked Boolean algebras and $\mathsf{S4.sBA}$ in \cite{Ina13}.

We summarize in Table \ref{tab:comparison} the three classes of structures and related intermediate (modal) logics.

\begin{table}[h!]
    \centering
        \scalebox{0.85}{
    \begin{tabular}{@{}lccccc@{}}
         \toprule
        \textbf{Name} & \textbf{Frame} & \textbf{Intermediate Logic} & \multicolumn{2}{c}{\textbf{Modal Logic}} \\
        \cmidrule(l){4-5}
        & & & $\tau$ & $\sigma$ \\
        \midrule
        Finite topless Boolean algebra & $\mathbf{P}_0(n)$ & $\mathsf{Med}$ & $\mathsf{S4.tBa}$ & $\mathsf{tBa}$ \\
        Finite spiked Boolean algebra & $\mathbf{S}(n)$ & $\mathsf{Inam}$ & $\mathsf{S4.sBa}$ & $\mathsf{sBa}$ \\
        Finite partial function algebra & $\mathbf{A}_n$ & $\mathsf{Cheq}$ & $\mathsf{S4.FPFA}$ & $\mathsf{L}_{\infty}$ \\
        \bottomrule
    \end{tabular}
    }
    \caption{Comparison of three classes of structures}
    \label{tab:comparison}
\end{table}

Let $p_0$, $p_1$, $p_2$, $q_0$, and $q_1$ be distinct propositional variables and define

\noindent
\(
\displaystyle
\begin{aligned}
&\varphi_0
:= (p_0 \wedge \neg p_1)\wedge\Box(p_0 \wedge \neg p_1),
\quad
\varphi_1
:= (p_1 \wedge \neg p_0)\wedge\Box(p_1 \wedge \neg p_0),
\\
&\phi_i
:= \bigl(p_i \wedge \bigwedge_{j \neq i}\neg p_j\bigr)
\wedge
\Box\bigl(p_i \wedge \bigwedge_{j \neq i}\neg p_j\bigr), 0\leq i \leq 2
\\
&\Psi_s
:= \bigl(
\Diamond(\Box q_0 \wedge \Diamond\varphi_0 \wedge \Diamond\varphi_1)
\wedge
\Diamond(\Box q_1 \wedge \Diamond\varphi_0 \wedge \Diamond\varphi_1)\bigr)
\rightarrow
\Diamond\Box(q_0 \wedge q_1),
\\
&\Psi_t
:= \bigl(\wedge_{i=0,1,2}\Diamond\phi_i\bigr)
\rightarrow
\Diamond(\Diamond\phi_0 \wedge \Diamond\phi_1 \wedge \Box\neg\phi_2)
\text{\footnotemark}.
\end{aligned}
\)
\footnotetext{This formula is due to \cite[\S~2]{HLL15}.}

\begin{lemma}
	$\Psi_s \in \mathsf{S4.sBa}$, $\Psi_s \notin \mathsf{S4.FPFA}$, $\Psi_s \notin \mathsf{S4.tBa}$.
\label{lem:psis}	
\end{lemma}
\begin{proof}
	Let $\mathbf{S}$ be any finite spiked pre-Boolean algebra and $V$ be a valuation on $\mathbf{S}$. Assume that a point $a \in \mathbf{S}$, $a \vDash \Diamond(\Box q_0 \wedge \Diamond \varphi_0 \wedge \Diamond \varphi_1) \wedge \Diamond(\Box q_1 \wedge \Diamond \varphi_0 \wedge \Diamond \varphi_1)$, then there are two nodes $x, y\geq a$ and $x\vDash \Box q_0 \wedge \Diamond \varphi_0 \wedge \Diamond \varphi_1$, $y \vDash \Box q_1 \wedge \Diamond \varphi_0 \wedge \Diamond \varphi_1$, so there are two points $u, v\geq x$, $u \vDash \varphi_0$ and $v \vDash \varphi_1$. By the definition of $\varphi_0$ and $\varphi_1$, we have $u \neq v$ and they belong to two different clusters. In fact, they have no common successor. It follows that $x$ is not a spike. The same argument applies to $y$. Thus $x$ and $y$ belong to the corresponding Boolean algebra, and then we have $z \in \mathbf{S}$ and $z \geq x, y$. Since $x \vDash \Box q_0$ and $y \vDash \Box q_1$, we have $z \vDash \Box (q_0\wedge q_1)$. Since $a\leq z$, then $a \vDash \Diamond\Box (q_0\wedge q_1)$. It follows that $\Psi_s \in \mathsf{S4.sBa}$.

Let $\mathbf{A}_2$ be the finite partial function algebra on two elements. $V'$ be a valuation on $\mathbf{A}_2$ such that
Let $V'$ be a valuation on $\mathbf{A}_2$ such that
\[
\begin{aligned}
(\{1,2\},\emptyset)
&\vDash \varphi_0 \wedge q_0 \wedge \neg q_1,
&\qquad
(\{2\},\{1\})
&\vDash \varphi_1 \wedge q_0 \wedge \neg q_1,\\
(\emptyset,\{1,2\})
&\vDash \varphi_1 \wedge \neg q_0 \wedge q_1,
&
(\{1\},\{2\})
&\vDash \varphi_0 \wedge \neg q_0 \wedge q_1,
\end{aligned}
\]
and
\[
(\{2\},\emptyset)\vDash q_0,
\qquad
(\emptyset,\{2\})\vDash q_1.
\]

Then $(\{2\}, \emptyset) \vDash \Box q_0 \wedge \Diamond \varphi_0\wedge \Diamond \varphi_1$ and $(\emptyset, \{2\}) \vDash \Box q_1 \wedge \Diamond \varphi_0\wedge \Diamond \varphi_1$. But $(\emptyset, \emptyset)\vDash \Box \Diamond \lnot (q_0 \wedge q_1)$, thus $\Psi_s \notin \mathsf{S4.FPFA}$. 

Let $\mathbf{P}_0(4)$ be the topless Boolean algebra on four elements. $V''$ be a valuation on $\mathbf{P}_0(4)$ such that
\[
\begin{aligned}
\{1\}
&\vDash \varphi_0 \wedge q_0 \wedge \neg q_1,
&\qquad
\{2\}
&\vDash \varphi_1 \wedge q_0 \wedge \neg q_1,\\
\{3\}
&\vDash \varphi_1 \wedge \neg q_0 \wedge q_1,
&
\{4\}
&\vDash \varphi_0 \wedge \neg q_0 \wedge q_1,
\end{aligned}
\]
and
\[
\{1,2\}\vDash q_0,
\qquad
\{3,4\}\vDash q_1.
\]

Then $\{1,2\} \vDash \Box q_0 \wedge \Diamond \varphi_0\wedge \Diamond \varphi_1$ and $\{3,4\} \vDash \Box q_1 \wedge \Diamond \varphi_0\wedge \Diamond \varphi_1$. But $\{1,2,3,4\}\vDash \Box \Diamond \lnot (q_0 \wedge q_1)$, thus $\Psi_s \notin \mathsf{S4.tBa}$. 
\end{proof}


\begin{lemma}
$\Psi_t \in \mathsf{S4.tBa}$, $\Psi_t \notin \mathsf{S4.FPFA}$, $\Psi_t \notin \mathsf{S4.sBa}$.	
\label{lem:psit}
\end{lemma}
\begin{proof}
	Let $\mathbf{P}_0$ be any finite topless pre-Boolean algebra corresponding to proper subsets of a finite set $D$ and $V$ be a valuation on $\mathbf{P}_0$. Assume that a point $a \in \mathbf{P}_0$ and $a \vDash \bigwedge_{i=0,1,2}\Diamond \phi_i$, then there are three nodes $c_0, c_1,c_2 \geq a$ such that $c_j \vDash \phi_j$ and no two of $c_j$ have a join in $\mathbf{P}_0$. Thus $c_0 \cap c_2 =c_1\cap c_2=\emptyset$, it follows that $a < c_0 \cup c_1 \nleq c_2$. Let $c=c_0 \cup c_1$, then $c \vDash \Diamond \phi_0 \wedge \Diamond \phi_1$. If $c \vDash \Diamond \phi_2$, then $c \vDash  \bigwedge_{i=0,1,2}\Diamond \phi_i$, and by the above argument, we can find a new $c' >  c$, and $c' \vDash \Diamond \phi_0 \wedge \Diamond \phi_1$. Because $\mathbf{P}_0$ is finite, we will finally found a node $c_0 >a$ such that $c_0 \vDash \Diamond \phi_0 \wedge\Diamond \phi_1 \wedge \Box \lnot  \phi_2$. It follows that $\Psi_t \in \mathsf{S4.tBa}$.
	
	Let $\mathbf{A}_2$ be a finite partial function algebra on two elements. $V'$ be a valuation on $\mathbf{A}_2$ such that 
\[
\begin{aligned}
(\{1,2\},\emptyset)
&\vDash \phi_0,
&\qquad
(\emptyset,\{1,2\})
&\vDash \phi_1,\\
(\{2\},\{1\})
&\vDash \phi_2,
&
(\{1\},\{2\})
&\vDash \phi_2.
\end{aligned}
\]
	
Then $(\emptyset, \emptyset) \vDash \bigwedge_{i=0,1,2}\Diamond \phi_i$. For any node $a\notin \{(\{1, 2\}, \emptyset),(\emptyset, \{1, 2\})$, we have $a\vDash \Diamond \phi_2$. For any node $a\in \{(\{1, 2\}, \emptyset),(\emptyset, \{1, 2\})$, we have $a\vDash \neg(\Diamond \phi_0 \wedge \Diamond\phi_1)$. Therefore, any point dose not have $\Diamond \phi_0 \wedge\Diamond \phi_1 \wedge \Box \lnot  \phi_2$ and so $\Psi_t \notin \mathsf{S4.FPFA}$. 

Let $\mathbf{S}(3)$ be the spiked Boolean algebra on three elements. $V''$ be a valuation on $\mathbf{S}(3)$ such that 
\[
\begin{aligned}
\{1s\}
&\vDash p_0 \wedge \neg p_1 \wedge \neg p_2,
&\qquad
\{2s\}
&\vDash \neg p_0 \wedge p_1 \wedge \neg p_2,\\
\{3s\}
&\vDash \neg p_0 \wedge \neg p_1 \wedge p_2,
&
\emptyset
&\vDash \neg p_0 \wedge \neg p_1 \wedge p_2.
\end{aligned}
\]

Then $\{1,2,3\} \vDash \bigwedge_{i=0,1,2}\Diamond \phi_i$. For any node $a\notin \{\{1s\},\{2s\}\}$, we have $a\vDash \Diamond \phi_2$. For any node $a\in \{\{1s\},\{2s\}\}$, we have $a\vDash \neg(\Diamond \phi_0 \wedge \Diamond\phi_1)$. Therefore, any point dose not have $\Diamond \phi_0 \wedge\Diamond \phi_1 \wedge \Box \lnot  \phi_2$ and so $\Psi_t \notin \mathsf{S4.sBa}$. 
\end{proof}
\begin{corollary}
 $\mathsf{S4.tBa}$ and $\mathsf{S4.sBa}$ do not contain each other.
\label{cor:notcontaineachother}
\end{corollary}

\begin{theorem}
$\mathsf{S4.FPFA}$ is strictly included in $\mathsf{S4.tBA}$.
\label{thm:fpfatba}
\end{theorem}

\begin{proof}
Since $\Psi_t \in \mathsf{S4.tBa}$ and $\Psi_t \notin \mathsf{S4.FPFA}$, it suffices to prove that for any $n \in \omega$, there exists a $p$-morphism $f_n$ from $\mathbf{A}_n$ onto $\mathbf{P}_0(n+1)$. Define a map $f_n$ from $\mathbf{A}_n$ to $\mathbf{P}_0(n+1)$ for any $a=(1_a,2_a)$ as follows:
\[ 
f_n(a)=
\begin{cases}
[v_a] \setminus 1_a, & \text{if } 2_a \neq \emptyset,\\
[n+1] \setminus 1_a, & \text{if } 2_a = \emptyset,
\end{cases}
\] 
where $v_a = \min\{v: v \in 2_a\}$.

We observe the following results:
\begin{enumerate}
	\item $f_n((\emptyset,\emptyset)) = [n+1]$.
	\item If $a = (\{u\}, \emptyset)$, then $f_n(a) = [n+1] \setminus \{u\}$, where $1 \leq u \leq n$.
	\item If $a = (\emptyset, \{v\})$, then $f_n(a) = [v]$, where $1 \leq v \leq n$.
\end{enumerate}

For any non-empty subset $S$ of $\{1,2,\ldots,n+1\}$, if $n+1 \in S$, then we have $f_n(([n+1] \setminus S,\emptyset)) = S$. If $n+1 \notin S$ but $n \in S$, then $f_n(([n] \setminus S,\{n\})) = S$. More generally, let $n_S = \max S$, then $f_n(([n_S] \setminus S,\{n_S\})) = S$. Thus, for each non-empty subset of $[n+1]$, it is the image of some element of $\mathbf{A}_n$, so $f_n$ is onto.

If $a \leq b$ in $\mathbf{A}_n$ and $a=(\emptyset,\emptyset)$, $b$ is an atom of $\mathbf{A}_n$, then $f_n(a) = [n+1] \supseteq f_n(b)$, so $f_n(a) \leq f_n(b)$ in $\mathbf{P}_0(n+1)$.

According to our construction, $f_n(a)= \bigcap f_n(x_a)$, where $x_a\leq a$  and $x_a$ is an atom in $\mathbf{A}_n$. Thus if $a \leq b$ in $\mathbf{A}_n$, then $f_n(a)=\bigcap f_n(x_a)\supseteq \bigcap f_n(x_b)=f_n(b)$ and so $f_n(a)\leq f_n(b)$ in $\mathbf{P}_0(n+1)$.

On the other hand, if $f_n(a) \leq s$ in $\mathbf{P}_0(n+1)$, then $s \subseteq f_n(a)$ as a non-empty subset of $[n+1]$. We aim to find $b$ in $\mathbf{A}_n$ such that $a \leq b$ and $f_n(b) = s$, there are three cases:

\emph{Case 1.} If $n+1 \in f_n(a)$ and $n+1 \in s$, then $2_a=\emptyset$ and $a=(1_a,\emptyset)$, $f_n(a)=[n+1]\setminus 1_a$. Let $1_b=\{u \in [n] : u \notin s\}$ and $b=(1_b, \emptyset)$. According to our construction, $f_n(b)=[n+1]\setminus 1_b= s$. Since $[n+1]\setminus 1_b=s\subseteq f_n(a)=[n+1]\setminus 1_a$, then $1_a\subseteq 1_b$ and so $a=(1_a,\emptyset)\leq (1_b,\emptyset)=b$.

\emph{Case 2.} If $n+1 \in f_n(a)$ and $n+1 \notin s$, then $a=(1_a,\emptyset)$ and $f_n(a)=[n+1]\setminus 1_a$. Let $1_b=\{u \in [n] : u \notin s\}$ and $n_s=\max s$. Then $1_b \cap \{n_s\}=\emptyset$ and let $b=(1_b, \{n_s\})$, so $f_n(b)=[n_s]\setminus1_b=s$. Since $[n_s]\setminus 1_b=s\subseteq f_n(a)=[n+1]\setminus 1_a$, then $u\notin 1_b \rightarrow u\notin 1_a \text{ when } u\leq n_s$ and $u\in 1_b$ when $n_s< u\leq n$, thus $1_a\subseteq 1_b$. Therefore, $a=(1_a,\emptyset)\leq (1_b,\{n_s\})=b$.

\emph{Case 3.} If $n+1 \notin f_n(a)$ and $n+1 \notin s$, then $2_a\neq \emptyset$ and $a=(1_a,2_a)$, $f_n(a)=[v_a]\setminus 1_a$. Let $n_s=\max s$ and $1_b=\{u \in [n_s] : u \notin s\} \cup \{u: n_s<u\leq n \text{ and } u\in 1_a\}$. Since $[n_s]\setminus 1_b=s\subseteq f_n(a)=[v_a]\setminus 1_a$, then $n_s=\text{max}\ s \leq \text{min} \{v : v\in 2_a\} =v_a$ and $u\notin 1_b \rightarrow u\notin 1_a \text{ when } u\in [n_s]$, thus $1_a \subseteq 1_b$. It is easy to see that $1_b \cap (2_a \cup\{n_s\})=(1_b\cap 2_a) \cup(1_b \cap \{n_s\})$ where $1_b \cap 2_a = (\{u \in [n_s] : u \notin s\} \cap 2_a) \cup (\{u : n_s < u \leq n \text{ and } u \in 1_a\} \cap 2_a) \text{ and } 1_b \cap \{n_s\} = \emptyset$. Since $n_s\leq \text{min}\{v: v\in 2_a\} \text{ and } n_s\in s$, then $\{u \in [n_s] : u \notin s\} \cap 2_a=\emptyset$. Because $\{u : n_s < u \leq n \text{ and } u \in 1_a\} \cap 2_a=\emptyset$, then $1_b\cap2_a=\emptyset$. Thus $1_b \cap (2_a \cup\{n_s\})=\emptyset$ and then we can let $b=(1_b, 2_a\cup \{n_s\})$. According to the above argument, $a=(1_a, 2_a)\leq (1_b,2_a\cup\{n_s\})=b$ and $f_n(b)=[\text{min}(2_a\cup \{n_s\})]\setminus 1_b=[n_s]\setminus1_b=s$.

Thus, $f_n$ is a $p$-morphism from $\mathbf{A}_n$ onto $\mathbf{P}_0(n+1)$.
\end{proof}

A different construction and proof for a \emph{p}-morphism between $\mathbf{V}_n$ and $\mathbf{P}_0(n+1)$ can be found in \cite{Lit04}.

\begin{theorem}
$\mathsf{S4.FPFA}$ is strictly included in $\mathsf{S4.sBA}$.
\label{thm:fpfasba}	
\end{theorem}

\begin{proof}
	Since $\Psi_s\in \mathsf{S4.sBa}$ and $\Psi_s \notin \mathsf{S4.FPFA}$, it is sufficient to get the result by proving that for any $n \in \omega$, there exists a $p$-morphism $g_n$ from $\mathbf{A}_n$ onto $\mathbf{S}(n)$, where \( \mathbf{S}(n)= \langle S_n, \leq_{S_n} \rangle \) denotes the spiked Boolean algebra on \( n \) elements.

We define a map $g_n$ from $\mathbf{A}_n$ to $\mathbf{S}(n)$ for any $a=(1_a,2_a)$ as follows:
\[
g_n(a)=\begin{cases}
	[n]\setminus1_a, & \text{ if } 2_a=\emptyset\\
	\{js\} , & \text{ if } a=([n]\setminus \{j\},\{j\})\\
	\bigcap_{j\in 2_a} \{j\}, & \text{ otherwise}\\
\end{cases}
\] 

It is not difficult to calculate the following cases:
\begin{enumerate}
	\item  $g_n((\emptyset,\emptyset ))=[n]$.
	\item If $a=(\{u\}, \emptyset)$, then $g_n(a)=[n]\setminus \{u\}$, where $1\leq u \leq n$.
	\item If $a=(\emptyset,\{v\})$, $g_n(a)=\{v\}$, where $1\leq v \leq n$.
	\item If $a=(\{1,2,...,n\}\setminus \{j\},\{j\})$, $g_n(a)=\{js\}$, where $1\leq j \leq n$.
\end{enumerate}

For any subset $S\subseteq [n]$, let $1_a=\{u\in [n] : u\notin S\}$, then $g_n((1_a, \emptyset))=[n]\setminus 1_a=S$. Together with the fact that $g_n(([n]\setminus \{j\},\{j\}))=\{js\}$, thus $g_n$ is an onto from $\mathbf{A}_n$ to $\mathbf{S}(n)$.

\emph{Case 1.} If $a\leq b$ in $\mathbf{A}_n$ and $2_b=\emptyset$, then $1_a \subseteq 1_b$ and so $g_n(a)=[n]\setminus 1_a \supseteq  [n]\setminus 1_b=g_n(b)$, that is, $g_n(a)\leq  g_n(b)$. 

\emph{Case 2.}  If $a\leq b$ in $\mathbf{A}_n$ and $b=(\{1,2,...,n\}\setminus \{j\},\{j\})$, then $a=(1_a, \emptyset)$ or $a=(1_a, \{j\})$ where $1_a\subseteq 1_b$. Thus $j\notin 1_a$ and so $j\in [n]\setminus 1_a$. Therefore, $g_n(a)=[n]\setminus 1_a \leq \{j\} < \{js\}=g_n(b) \text{ if } a=(1_a, \emptyset)$ or $g_n(a)=\{j\} < \{js\}=g_n(b) \text{ if } a=(1_a, \{j\})$. So $g_n(a)\leq g_n(b)$ in $\mathbf{S}(n)$.

\emph{Case 3.}  If $a\leq b$ in $\mathbf{A}_n$ and $b\neq (\{1,2,...,n\}\setminus \{j\},\{j\})$ for any $1\leq j\leq n$ and $2_b\neq \emptyset$. So there exists at least a $j$ such that $j\in 2_b$. Then $j\notin 1_b$ and so $j\notin 1_a$. Thus $g_n(a)=[n]\setminus 1_a \leq \{j\} \leq \bigcap_{i\in 2_b}\{i\}=g_n(b) \text{ if } 2_a= \emptyset$ or $g_n(a)= \bigcap_{i\in 2_a}\{i\} \leq \bigcap_{i\in 2_b}\{i\}=g_n(b) \text{ if } 2_a\neq \emptyset$. So $g_n(a)\leq g_n(b)$ in $\mathbf{S}(n)$.

On the other hand, if $g_n(a)\leq s$ in $\mathbf{S}(n)$. We aim to find $b$ in $\mathbf{A}_n$, such that $a\leq b$ and $g_n(b)=s$. 

If $2_a=\emptyset$, then $g_n(a)=[n]\setminus 1_a$ and $a=(1_a, \emptyset)$. There are four cases for $s$ as follows:

\emph{Case 1.}  If $s$ is a non-empty subset of $[n]$ and $s\neq \{j\}$ for any $1\leq j \leq n$. Let $b=([n]\setminus s, \emptyset)$, then $g_n(b)=s$, and since $g_n(a)=[n]\setminus 1_a \supseteq s$, then $[n]\setminus s\supseteq 1_a$, therefore, $a=(1_a, \emptyset)\leq([n]\setminus s, \emptyset)=b$.	

\emph{Case 2.} If $s=\emptyset$, then let $b=([n], \emptyset)$. It is not difficult to see $g_n(b)=\emptyset=s$ and $a=(1_a, \emptyset)\leq([n], \emptyset)=b$.

\emph{Case 3.}  If $s=\{j\}$, since $[n]\setminus 1_a \leq \{j\}$, then $j\notin 1_a$. If $1_a=[n]\setminus \{j\}$, then $g_n(a)=\{j\}=s$, let $b=a$. If $1_a \neq [n]\setminus \{j\}$, Let $b=(1_a, \{j\})$, then $g_n(b)=\{j\}=s$ and $a=(1_a, \emptyset)\leq(1_a, \{j\})=b$.

\emph{Case 4.} If $s=\{js\}$, let $b=([n]\setminus \{j\}, \{j\})$. Since $[n]\setminus 1_a \leq \{js\}$, then $j\notin 1_a$ and so $1_a \subseteq [n]\setminus\{j\}$. Therefore, $a=(1_a, \emptyset)\leq([n]\setminus\{j\}, \{j\})=b$ and $g_n(b)=\{js\}=s$.

When $2_a\neq \emptyset$, we are looking at two obvious cases first. If $a= ([n]\setminus \{j\}, \{j\})$, then $g_n(a)=\{js\}$ and so $s=\{js\}$ and we choose $b=a$. Another case is that if $|2_a|\geq2$, then $g_n(a)=\emptyset$ and so $s=\emptyset$, then we choose $b=a$.

The remaining case is that if $2_a=\{j\}$ and $a\neq ([n]\setminus \{i\}, \{i\})$ for any $1\leq i\leq n$, then $g_n(a)=\{j\}$. Let $a=(1_a, \{j\})$.

\emph{Case 1.}  If $s=\{js\}$, since $1_a\neq [n]\setminus\{j\}$, then $1_a \subseteq [n]\setminus \{j\}$ and let $b=( [n]\setminus \{j\},\{j\})$, then $g_n(b)=\{js\}=s$ and $a=(1_a,\{j\})\leq ( [n]\setminus \{j\},\{j\})=b$.

\emph{Case 2.}  If $s=\emptyset$, since $1_a\neq [n]\setminus\{j\}$, then there exists $u\neq j$ such that $u \in [n]\setminus 1_a$, thus $1_a\cap\{u,j\}=\emptyset$. Let $b=(1_a, \{u,j\})$, then $g_n(b)=\emptyset=s$ and $a=(1_a, \{j\})\leq(1_a, \{u,j\})=b$.

Thus, $g_n$ is a $p$-morphism from $\mathbf{A}_n$ onto $\mathbf{S}(n)$. 
\end{proof}
\begin{corollary}
$\mathsf{L}_{\infty} \subseteq \mathsf{tBA}\cap \mathsf{sBA}$, $\mathsf{Cheq} \subseteq \mathsf{Med}\cap \mathsf{Inam}$, and $\mathsf{S4.FPFA} \subseteq \mathsf{S4.tBA}\cap \mathsf{S4.sBA}$. 
\label{cheqsubsetmedcapls}
\end{corollary}

Let $\mathbf{W}$ denote the \emph{bowtie}, the poset consisting 
of four elements: two incomparable minimal elements and two 
incomparable maximal elements, each maximal element lying above 
each minimal one; see Figure~\ref{fig:bowtie}. 
\begin{figure}[htbp]
$$\xymatrix@C=5mm{
\bullet\ar@{-}[d]\ar@{-}[rrd] & & \bullet\ar@{-}[d]\\
\bullet\ar@{-}[rru] & & \bullet\\
}$$
\caption{The bowtie $\mathbf{W}$.\label{fig:bowtie}}
\end{figure}
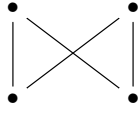

\begin{lemma}[Bowtie Lemma]	
If $\mathbf{F}$ is any finite rooted frame, then 
$\mathbf{F}\oplus\mathbf{W}$ is a $p$-morphic image of some 
$\mathbf{V}_n$.
\label{lem:bow-tie}
\end{lemma}
The proof of Bowtie Lemma can be found in~\mbox{\cite[Lemma~6.2.3]{Xiao24}}. The Bowtie Lemma provides a useful tool for the analysis of finite $\mathsf{Cheq}$-frames. 
Fontaine proved that $\mathsf{Med}$ in not finitely axiomtizable over $\mathsf{Cheq}$ in \cite[Theorem~9]{Fon06}. By using Bowtie Lemma, analogous results for the the Inamdar logic $\mathsf{Inam}$ and for generalized Medvedev logics were obtained in \cite[Theorem~18]{LX24} and \cite[Theorem~5.3]{Xiao26}, respectively.

\section{Geometric perspective}
\label{sec:geo}
\begin{definition}
Assuming $X$ is a polytope in $\mathbb{R}^{n}$, a hyperplane $c$ of $\mathbb{R}^{n}$ is \emph{supporting} $X$ if one of the two closed half-spaces of $c$ contains $X$. A subset $F$ of $X$ is called a \emph{face} of $X$ if it is either $\emptyset$, $X$ itself or the intersection of $X$ with a supporting hyperplane $c$. The empty set $\emptyset$ and the $X$ itself are defined to be \emph{trivial faces}. Otherwise, the face is called \emph{proper}. The maximal proper face under the inclusion is called \emph{facet} while the minimal proper face is called \emph{vertex}. The \emph{dual face poset} $\mathbf{Q}^{\text{d}}(X)$ is the poset of non-empty faces of $X$, ordered by reverse inclusion: $F_0\leq F_1$ \mbox{iff $F_0 \supseteq F_1$}.	 
\end{definition}

If we add the top $\emptyset$ to $\mathbf{Q}^{\text{d}}(X)$, then it forms a bounded lattice $\mathbf{Q}_1^{\text{d}}(X)$ since $F_0\vee F_1=F_0\cap F_1$ and $F_0\wedge F_1=\bigcap_{(F_0\cup F_1)\subseteq F} F$. We call $\mathbf{Q}_1^{\text{d}}(X)$ the \emph{face lattice} of $X$. Note that, in the lattice theory, the face lattice of $X$ is the opposite lattice of $\mathbf{Q}_1^{\text{d}}(X)$, however, we will continue to use our terminology, and this should not cause any confusion.
For the convenience of the discussion, we use the coordinate system to give a definition of the $n$-cube $Q_n$.
\begin{definition}
An $n$-\emph{cube} $Q_n$ is the set $\{(x_1, ..., x_n) \in \mathbb{R}^n : 1\leq x_i \leq2 \text{ for } 1\leq i\leq n\}$.
\end{definition}

Bennett \cite[Theorem 2.1]{Ben82} provided a description of the face lattice and we state the result below.

\begin{theorem}
A lattice $\mathbf{F}$ is the face lattice of $Q_n$ iff $|\mathbf{F}|=3^n+1$, $|\mathrm{at}(\mathbf{F})|=2n$, $\mathbf{F}$ is atomistic and for each atom $a$, there exists a unique atom $a'$ with $a \vee a'=\bar 1$. 
\label{theorem:facelattice}    	
\end{theorem}

\begin{theorem}
$\mathbf{V}_n \cong \mathbf{Q}^{\text{d}}(Q_n)$, the dual face poset of $n$-cube. 
\label{thm:dualfaceposetofcube}
\end{theorem}
\begin{proof}
Let $\mathbf{V}_n'$ be a frame obtained by adding a top $\bar{1}$ to $\mathbf{V}_n$. The cardinality of $\mathbf{V}_n'$ is $3^n+1$ and $\mathbf{V}_n'$ has $2n$ atoms $\{x^i\}_{1\leq i\leq n}$ and $\{y^i\}_{1\leq i\leq n}$ where  $x^i=(\underbrace{0, \ldots ,0}_{i-1}, 1 , 0, \ldots, 0)$ and $y^i=(\underbrace{0, \ldots , 0}_{i-1}, 2, 0, \ldots, 0)$. For any $1\leq i\neq j \leq n$, it is easy to see $x^i \vee y^i=\bar 1$. Since $x^i \vee y^j$, $x^i \vee x^j$ and $y^i \vee y^j$ are all in $\mathbf{V}_n$, thus not $\bar 1$. For any point $a\neq \bar 1$ in $\mathbf{V}_n'$, let $a=(a_1, \ldots, a_n)$, then
$a= (\bigvee_{i \in \pi_1} x_i)\vee(\bigvee_{j \in \pi_2} y_j)$, where $\pi_1=\{i : a_i=1\} \text{ and } \pi_2=\{i : a_i=2\}$.

Therefore $\mathbf{V}_n'\cong \mathbf{Q}_1^{\text{d}}(Q_n)$ by Theorem \ref{theorem:facelattice}. By removing the top $\bar{1}$, $\mathbf{V}_n \cong \mathbf{Q}^{\text{d}}(Q_n)$.  
\end{proof}
Similarly, Medvedev frames also have a natural geometric structure.
\begin{definition}
	A $n$-\emph{simplex} $X_n$ is a convex hull of $n+1$ affinely independent points.
\end{definition}

\begin{theorem}
	$\mathbf{P}_0(n) \cong \mathbf{Q}^{\text{d}}(X_n)$, the dual face poset of $n$-simplex. 
\end{theorem}

\subsection{Polychromatic colorings of \emph{n}-cube}
The \emph{n}-dimensional hypercube $Q_n$ is the graph whose vertex set is $\{1, 2\}^{n}$. Let $E_n$ denote the set of all edges in $Q_n$. The two vertices $v$ and $w$ of an edge $e$ of the hypercube differ in exactly one coordinate, thus we represent the edge $e\in E_n$ by a $n$-vector $(e_1,e_2,\ldots,e_n)$ with
\[e_i=\begin{cases}
0, & \text{ if } v \text{ and } w \text{ differ in } i \text{-th coordinate}\\
v_i, & \text{ if } v_i=w_i  
\end{cases}
\]

Let $G_0$ and $G_1$ be two graphs. A subgraph of $G_0$ isomorphic to $G_1$ is called an \emph{embedding} of $G_1$ in $G_0$. Furthermore, an embedding of $Q_d$ in $Q_n$ (as a subcube) is a $n$-vector with $d$ coordinates set to $0$. 

Fix a set $P$ of $p$ colors. An \emph{edge}-\emph{coloring} $\mu$ of a graph $G$ with $p$ colors is a surjective function 
\[
\mu: E(G) \rightarrow P
\]
which associate every edge of $G$ with a color in $P$.
  
For a given hypercube $Q_n$ and a fixed subgraph $G$, an edge-coloring $\mu$ of a hypercube with $p$ colors is called a $G$-\emph{polychromatic} $p$-\emph{coloring of} $Q_n$ if every embedding of $G$ in $Q_n$ contains every color. Let $p(G)$ be the maximum number of colors for which a $G$-\emph{polychromatic coloring} is possible for the edges of any hypercube (containing $G$). $p(G)$ is called the \emph{polychromatic number} of $G$.  

The case $G=Q_d$, a sub-hypercube in $Q_n$, was introduced by Alon, Krech and Szab{\'o} in \cite[Theorem 4]{AKS07}.      
\begin{theorem}[Alon, Krech and Szab{\'o}]
\[
\displaystyle\binom{d+1}{2} \geq p(Q_d) \geq \left\lfloor \frac{(d+1)^2}{4} \right\rfloor
\]
\end{theorem}

Offner gave the exact value of the polychromatic number of $Q_d$ in \cite[Theorem 2]{Off08}.

\begin{theorem}[Offner] 

\[
p(Q_d)=\left\lfloor \frac{(d+1)^2}{4} \right\rfloor
\]
\label{thm:offner}
\end{theorem}

By using Offner's result \ref{thm:offner}, Kuznetsov in \cite{Kuz19} discussed the two strategies belonging to Fontaine and Shatrov for proving that $\mathsf{Cheq}$ is not finitely axiomatizable, respectively. He also presented an alternative solution as an open problem for proving the non-finitely axiomatizability of $\mathsf{Cheq}$.

\subsection{Fontaine's structures and Shatrov's structures}
Recall that Maksimova, Skvortsov and Shehtman provided a method to prove that $\mathsf{Med}$ is not finitely axiomatizable in \cite[Corollary 5]{MSS79} by showing $\mathbf{L}_{s,2^{s+3}}$ is not $\mathsf{Med}$-frame while $\mathbf{L}'_{s,2^{s+3},m}$ is $\mathsf{Med}$-frame for any $s\geq 1$.   

\begin{figure}[H]
\centering
\begin{tikzpicture}[scale=0.3, node distance=2cm, every node/.style={circle, fill=black, minimum size=3.5pt, inner sep=0pt}]
\begin{scope}[xshift=-4cm] 
    
    \node (1)[label={[font=\small]below:\text{($s$+2,0)}}]at (0,0) {};
    \node (2)[label={[font=\scriptsize]left:($s$+1,0)}] at (-4.5,2.5) {};
    \node (3) at (-2,2.5) {};
    \node (4) at (2,2.5) {};
    \node (5)[label={[font=\scriptsize]right:($s$+1,$2\times3^{s+2}$)}] at (4.5,2.5) {};
    \node (6)[label={[font=\small]left:($s$,0)}] at (-1.7,5) {};
    \node (7)[label={[font=\small]right:($s$,2)}] at (1.7,5) {};
    \node (8)[label={[font=\small]left:($s$-1,0)}] at (-1.7,6.5) {};
    \node (9)[label={[font=\small]right:($s$-1,2)}] at (1.7,6.5) {};
    \node (10)[label={[font=\small]left:($m$+1,0)}] at (-1.7,9) {};
    \node (11)[label={[font=\small]right:($m$+1,2)}] at (1.7,9) {};
    \node (12)[label={[font=\small]left:($m$,0)}] at (-1.7,10.5) {};
    \node (13)[label={[font=\small]right:($m$,2)}] at (1.7,10.5) {};
    \node (14)[label={[font=\small]left:($m$-1,0)}] at (-1.7,12) {};
    \node (15)[label={[font=\small]right:($m$-1,2)}] at (1.7,12) {};

   \node (16)[label={[font=\small]left:(3,0)}] at (-1.7,14.5) {};
   \node (17)[label={[font=\small]right:(3,2)}] at (1.7,14.5) {};
   \node (18)[label={[font=\small]left:(2,0)}] at (-1.7,16) {};
   \node (19)[label={[font=\small]right:(2,2)}] at (1.7,16) {};
    
   \node (20)[label={[font=\small]left:(1,0)}] at (-1.7,17.5) {};
   \node (21)[label={[font=\small]right:(1,2)}] at (1.7,17.5) {};
   \node (22)[label={[font=\small]left:(0,0)}] at (-1.7,19) {};
   \node (23)[label={[font=\small]right:(0,1)}] at (1.7,19) {};
  
   \node (31) at (0,17.5) {};

 \node (24) at (0,16) {};
 \node (25) at (0,14.5) {};
  \node (26) at (0,12) {};
  \node (27) at (0,10.5) {};
  \node (28) at (0,9) {};
  \node (29) at (0,6.5) {};
  \node (30) at (0,5) {};

\draw[thin, dashed, dash pattern=on 1.6pt off 1pt] (8) -- ++(0,0.6); 
\draw[thin, dashed, dash pattern=on 1.6pt off 1pt] (9) -- ++(0,0.6); 
\draw[thin, dashed, dash pattern=on 1.6pt off 1pt] (10) -- ++(0,-0.6); 
\draw[thin, dashed, dash pattern=on 1.6pt off 1pt] (11) -- ++(0,-0.6);     

\draw[thin, dashed, dash pattern=on 1.6pt off 1pt] (29) -- ++(0,0.6); 
\draw[thin, dashed, dash pattern=on 1.6pt off 1pt] (28) -- ++(0,-0.6);

\draw[thin, dashed, dash pattern=on 1.6pt off 1pt] (14) -- ++(0,0.6); 
\draw[thin, dashed, dash pattern=on 1.6pt off 1pt] (15) -- ++(0,0.6); 
\draw[thin, dashed, dash pattern=on 1.6pt off 1pt] (16) -- ++(0,-0.6); 
\draw[thin, dashed, dash pattern=on 1.6pt off 1pt] (17) -- ++(0,-0.6);     

\draw[thin, dashed, dash pattern=on 1.6pt off 1pt] (26) -- ++(0,0.6); 
\draw[thin, dashed, dash pattern=on 1.6pt off 1pt] (25) -- ++(0,-0.6);

   \fill (-0.5,2.5) circle (1.5pt);
    \fill (-1,2.5) circle (1.5pt);
    \fill (-0.75,2.5) circle (1.5pt);
  \fill (1,2.5) circle (1.5pt);
    \fill (0.75,2.5) circle (1.5pt);
    \fill (0.5,2.5) circle (1.5pt);

    \draw[thick] (1) -- (2);
    \draw[thick] (1) -- (3);
    \draw[thick] (1) -- (4);
    \draw[thick] (1) -- (5);
    
    \draw[thick] (2) -- (6);
    \draw[thick] (2) -- (7);
     \draw[thick] (3) -- (6);
    \draw[thick] (3) -- (7);
     \draw[thick] (4) -- (6);
    \draw[thick] (4) -- (7);
     \draw[thick] (5) -- (6);
    \draw[thick] (5) -- (7); 
    
    \draw[thick] (6) -- (8);
    \draw[thick] (6) -- (9);
    \draw[thick] (7) -- (8);
    \draw[thick] (7) -- (9);

    \draw[thick] (10) -- (12);
    \draw[thick] (10) -- (13);
    \draw[thick] (11) -- (12);
    \draw[thick] (11) -- (13);
    \draw[thick] (12) -- (14);
    \draw[thick] (12) -- (15);
    \draw[thick] (13) -- (14);
    \draw[thick] (13) -- (15);
   
    \draw[thick] (16) -- (18);
    \draw[thick] (16) -- (19);
    \draw[thick] (17) -- (18);
    \draw[thick] (17) -- (19);

\draw[thick] (18) -- (20);
    \draw[thick] (18) -- (21);
    \draw[thick] (19) -- (20);
    \draw[thick] (19) -- (21);
\draw[thick] (20) -- (22);
    \draw[thick] (20) -- (23);
    \draw[thick] (21) -- (22);
    \draw[thick] (21) -- (23);

\draw[thick] (20) -- (24);
    \draw[thick] (21) -- (24);
    \draw[thick] (16) -- (24);
    \draw[thick] (25) -- (24);
\draw[thick] (17) -- (24);
    \draw[thick] (25) -- (18);
    \draw[thick] (25) -- (19);

    \draw[thick] (26) -- (12);
\draw[thick] (26) -- (27);
    \draw[thick] (26) -- (13);
 
 \draw[thick] (27) -- (14);
\draw[thick] (15) -- (27);
    \draw[thick] (28) -- (13); 
    \draw[thick] (28) -- (12);
\draw[thick] (10) -- (27);
    \draw[thick] (11) -- (27);
    
     \draw[thick] (29) -- (6);
\draw[thick] (29) -- (7);
    \draw[thick] (29) -- (30); 
    \draw[thick] (8) -- (30);
    
     \draw[thick] (9) -- (30);
\draw[thick] (2) -- (30);
    \draw[thick] (3) -- (30); 
    \draw[thick] (4) -- (30);
    \draw[thick] (5) -- (30);
    \draw[thick] (27) -- (28);

 \draw[thick] (31) -- (23);
 \draw[thick] (31) -- (22);
 \draw[thick] (31) -- (18);
  \draw[thick] (31) -- (19);
  \draw[thick] (31) -- (24);

\end{scope}

\begin{scope}[xshift=13.9cm] 
    
    \node (1)[label={[font=\small]below:\text{($s$+2,0)}}]at (0,0) {};
    \node (2)[label={[font=\scriptsize]left:($s$+1,0)}] at (-4.5,2.5) {};
    \node (3) at (-2,2.5) {};
    \node (4) at (2,2.5) {};
    \node (5)[label={[font=\scriptsize]right:($s$+1,$2\times3^{s+2}$)}] at (4.5,2.5) {};
    \node (6)[label={[font=\small]left:($s$,0)}] at (-1.7,5) {};
    \node (7)[label={[font=\small]right:($s$,2)}] at (1.7,5) {};
    \node (8)[label={[font=\small]left:($s$-1,0)}] at (-1.7,6.5) {};
    \node (9)[label={[font=\small]right:($s$-1,2)}] at (1.7,6.5) {};
    \node (10)[label={[font=\small]left:($m$+1,0)}] at (-1.7,9) {};
    \node (11)[label={[font=\small]right:($m$+1,2)}] at (1.7,9) {};
   
    \node (14)[label={[font=\small]left:($m$-1,0)}] at (-1.7,12) {};
    \node (15)[label={[font=\small]right:($m$-1,2)}] at (1.7,12) {};

   \node (16)[label={[font=\small]left:(3,0)}] at (-1.7,14.5) {};
   \node (17)[label={[font=\small]right:(3,2)}] at (1.7,14.5) {};
   \node (18)[label={[font=\small]left:(2,0)}] at (-1.7,16) {};
   \node (19)[label={[font=\small]right:(2,2)}] at (1.7,16) {};
    
   \node (20)[label={[font=\small]left:(1,0)}] at (-1.7,17.5) {};
   \node (21)[label={[font=\small]right:(1,2)}] at (1.7,17.5) {};
   \node (22)[label={[font=\small]left:(0,0)}] at (-1.7,19) {};
   \node (23)[label={[font=\small]right:(0,1)}] at (1.7,19) {};
  
   \node (31) at (0,17.5) {};

 \node (24) at (0,16) {};
 \node (25) at (0,14.5) {};
  \node (26) at (0,12) {};

 \node (27)[label={[font=\small,label distance=1.05cm]180:($m$,0)}] at (0,10.5) {};

  \node (28) at (0,9) {};
  \node (29) at (0,6.5) {};
  \node (30) at (0,5) {};

\draw[thin, dashed, dash pattern=on 1.6pt off 1pt] (8) -- ++(0,0.6); 
\draw[thin, dashed, dash pattern=on 1.6pt off 1pt] (9) -- ++(0,0.6); 
\draw[thin, dashed, dash pattern=on 1.6pt off 1pt] (10) -- ++(0,-0.6); 
\draw[thin, dashed, dash pattern=on 1.6pt off 1pt] (11) -- ++(0,-0.6);     

\draw[thin, dashed, dash pattern=on 1.6pt off 1pt] (29) -- ++(0,0.6); 
\draw[thin, dashed, dash pattern=on 1.6pt off 1pt] (28) -- ++(0,-0.6);

\draw[thin, dashed, dash pattern=on 1.6pt off 1pt] (14) -- ++(0,0.6); 
\draw[thin, dashed, dash pattern=on 1.6pt off 1pt] (15) -- ++(0,0.6); 
\draw[thin, dashed, dash pattern=on 1.6pt off 1pt] (16) -- ++(0,-0.6); 
\draw[thin, dashed, dash pattern=on 1.6pt off 1pt] (17) -- ++(0,-0.6);     

\draw[thin, dashed, dash pattern=on 1.6pt off 1pt] (26) -- ++(0,0.6); 
\draw[thin, dashed, dash pattern=on 1.6pt off 1pt] (25) -- ++(0,-0.6);

   \fill (-0.5,2.5) circle (1.5pt);
    \fill (-1,2.5) circle (1.5pt);
    \fill (-0.75,2.5) circle (1.5pt);
  \fill (1,2.5) circle (1.5pt);
    \fill (0.75,2.5) circle (1.5pt);
    \fill (0.5,2.5) circle (1.5pt);

    \draw[thick] (1) -- (2);
    \draw[thick] (1) -- (3);
    \draw[thick] (1) -- (4);
    \draw[thick] (1) -- (5);
    
    \draw[thick] (2) -- (6);
    \draw[thick] (2) -- (7);
     \draw[thick] (3) -- (6);
    \draw[thick] (3) -- (7);
     \draw[thick] (4) -- (6);
    \draw[thick] (4) -- (7);
     \draw[thick] (5) -- (6);
    \draw[thick] (5) -- (7); 
    
    \draw[thick] (6) -- (8);
    \draw[thick] (6) -- (9);
    \draw[thick] (7) -- (8);
    \draw[thick] (7) -- (9);

    \draw[thick] (16) -- (18);
    \draw[thick] (16) -- (19);
    \draw[thick] (17) -- (18);
    \draw[thick] (17) -- (19);

\draw[thick] (18) -- (20);
    \draw[thick] (18) -- (21);
    \draw[thick] (19) -- (20);
    \draw[thick] (19) -- (21);
\draw[thick] (20) -- (22);
    \draw[thick] (20) -- (23);
    \draw[thick] (21) -- (22);
    \draw[thick] (21) -- (23);

\draw[thick] (20) -- (24);
    \draw[thick] (21) -- (24);
    \draw[thick] (16) -- (24);
    \draw[thick] (25) -- (24);
\draw[thick] (17) -- (24);
    \draw[thick] (25) -- (18);
    \draw[thick] (25) -- (19);

\draw[thick] (26) -- (27);
 
 \draw[thick] (27) -- (14);
\draw[thick] (15) -- (27);
\draw[thick] (10) -- (27);
    \draw[thick] (11) -- (27);
    
     \draw[thick] (29) -- (6);
\draw[thick] (29) -- (7);
    \draw[thick] (29) -- (30); 
    \draw[thick] (8) -- (30);
    
     \draw[thick] (9) -- (30);
\draw[thick] (2) -- (30);
    \draw[thick] (3) -- (30); 
    \draw[thick] (4) -- (30);
    \draw[thick] (5) -- (30);
    \draw[thick] (27) -- (28);

 \draw[thick] (31) -- (23);
 \draw[thick] (31) -- (22);
 \draw[thick] (31) -- (18);
  \draw[thick] (31) -- (19);
  \draw[thick] (31) -- (24);

\end{scope}

\end{tikzpicture}

\caption{The frames $\mathbf{G}_s$ and $\mathbf{G}'_{s,m}$}
\label{fig:gaelle}
\end{figure}
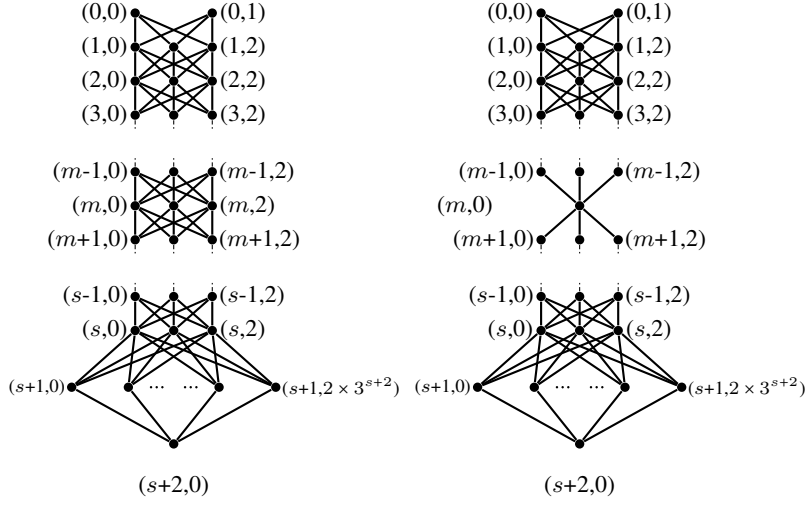

\begin{figure}[h]
\centering
\begin{tikzpicture}[scale=0.3, node distance=2cm, every node/.style={circle, fill=black, minimum size=3.5pt, inner sep=0pt}]

\begin{scope}[xshift=-4cm] 
    
    \node (1)[label={[font=\small]below:\text{($s$+2,0)}}]at (0.85,1) {};
    \node (2)[label={[font=\scriptsize]left:($s$+1,0)}] at (-2.8,3) {};
    \node (3) at (-0.85,3) {};
    \node (4) at (2.55,3) {};
    \node (5)[label={[font=\scriptsize]right:($s$+1,$2^{s+3}$)}] at (4.5,3) {};
    
    \node (6)[label={[font=\small]left:($s$,0)}] at (-1.7,5) {};
    \node (7)[label={[font=\small]right:($s$,3)}] at (3.4,5) {};
    \node (8)[label={[font=\small]left:($s$-1,0)}] at (-1.7,6.5) {};
    \node (9)[label={[font=\small]right:($s$-1,3)}] at (3.4,6.5) {};
   
    \node (10)[label={[font=\small]left:($m$+1,0)}] at (-1.7,9) {};
    \node (11)[label={[font=\small]right:($m$+1,3)}] at (3.4,9) {};
    \node (12)[label={[font=\small]left:($m$,0)}] at (-1.7,10.5) {};
    \node (13)[label={[font=\small]right:($m$,3)}] at (3.4,10.5) {};
    \node (14)[label={[font=\small]left:($m$-1,0)}] at (-1.7,12) {};
    \node (15)[label={[font=\small]right:($m$-1,3)}] at (3.4,12) {};

   \node (16)[label={[font=\small]left:(2,0)}] at (-1.7,14.5) {};
   \node (17)[label={[font=\small]right:(2,3)}] at (3.4,14.5) {};
   \node (18)[label={[font=\small]left:(1,0)}] at (-1.7,16) {};
   \node (19)[label={[font=\small]right:(1,3)}] at (3.4,16) {};
    
   \node (20)[label={[font=\small]left:(0,0)}] at (0,17.5) {};
   \node (21)[label={[font=\small]right:(0,1)}] at (1.7,17.5) {};

 \node (24) at (0,16) {};
 \node (25) at (0,14.5) {};
  \node (26) at (0,12) {};
  \node (27) at (0,10.5) {};
  \node (28) at (0,9) {};
  \node (29) at (0,6.5) {};
  \node (30) at (0,5) {};

 \node (31) at (1.7,16) {};
 \node (32) at (1.7,14.5) {};
  \node (33) at (1.7,12) {};
  \node (34) at (1.7,10.5) {};
  \node (35) at (1.7,9) {};
 \node (36) at (1.7,6.5) {};
  \node (37) at (1.7,5) {};

\draw[thin, dashed, dash pattern=on 1.6pt off 1pt] (8) -- ++(0,0.6); 
\draw[thin, dashed, dash pattern=on 1.6pt off 1pt] (9) -- ++(0,0.6); 
\draw[thin, dashed, dash pattern=on 1.6pt off 1pt] (10) -- ++(0,-0.6); 
\draw[thin, dashed, dash pattern=on 1.6pt off 1pt] (11) -- ++(0,-0.6);     

\draw[thin, dashed, dash pattern=on 1.6pt off 1pt] (29) -- ++(0,0.6); 
\draw[thin, dashed, dash pattern=on 1.6pt off 1pt] (28) -- ++(0,-0.6); 

\draw[thin, dashed, dash pattern=on 1.6pt off 1pt] (35) -- ++(0,-0.6); 
\draw[thin, dashed, dash pattern=on 1.6pt off 1pt] (36) -- ++(0,0.6);

\draw[thin, dashed, dash pattern=on 1.6pt off 1pt] (14) -- ++(0,0.6); 
\draw[thin, dashed, dash pattern=on 1.6pt off 1pt] (15) -- ++(0,0.6); 
\draw[thin, dashed, dash pattern=on 1.6pt off 1pt] (16) -- ++(0,-0.6); 
\draw[thin, dashed, dash pattern=on 1.6pt off 1pt] (17) -- ++(0,-0.6);     

\draw[thin, dashed, dash pattern=on 1.6pt off 1pt] (26) -- ++(0,0.6); 
\draw[thin, dashed, dash pattern=on 1.6pt off 1pt] (25) -- ++(0,-0.6); 
\draw[thin, dashed, dash pattern=on 1.6pt off 1pt] (32) -- ++(0,-0.6); 
\draw[thin, dashed, dash pattern=on 1.6pt off 1pt] (33) -- ++(0,0.6);     
       
   \fill (-0.05,3) circle (1.5pt);
   \fill (0.2,3) circle (1.5pt);
   \fill (0.45,3) circle (1.5pt);
   
   \fill (1.75,3) circle (1.5pt);
   \fill (1.5,3) circle (1.5pt);
   \fill (1.25,3) circle (1.5pt);

    \draw[thick] (1) -- (2);
    \draw[thick] (1) -- (3);
    \draw[thick] (1) -- (4);
    \draw[thick] (1) -- (5);
    
    \draw[thick] (2) -- (6);
    \draw[thick] (2) -- (7);
     \draw[thick] (3) -- (6);
    \draw[thick] (3) -- (7);
     \draw[thick] (4) -- (6);
    \draw[thick] (4) -- (7);
     \draw[thick] (5) -- (6);
    \draw[thick] (5) -- (7); 
    
    \draw[thick] (6) -- (8);
    \draw[thick] (6) -- (9);
    \draw[thick] (7) -- (8);
    \draw[thick] (7) -- (9);

    \draw[thick] (10) -- (12);
    \draw[thick] (10) -- (13);
    \draw[thick] (11) -- (12);
    \draw[thick] (11) -- (13);
    \draw[thick] (12) -- (14);
    \draw[thick] (12) -- (15);
    \draw[thick] (13) -- (14);
    \draw[thick] (13) -- (15);
   
    \draw[thick] (16) -- (18);
    \draw[thick] (16) -- (19);
    \draw[thick] (17) -- (18);
    \draw[thick] (17) -- (19);

\draw[thick] (18) -- (20);
    \draw[thick] (18) -- (21);
    \draw[thick] (19) -- (20);
    \draw[thick] (19) -- (21);

\draw[thick] (20) -- (24);
    \draw[thick] (21) -- (24);
    \draw[thick] (16) -- (24);
    \draw[thick] (25) -- (24);
\draw[thick] (17) -- (24);
    \draw[thick] (25) -- (18);
    \draw[thick] (25) -- (19);

    \draw[thick] (26) -- (12);
\draw[thick] (26) -- (27);
    \draw[thick] (26) -- (13);
 
 \draw[thick] (27) -- (14);
\draw[thick] (15) -- (27);
    \draw[thick] (28) -- (13); 
    \draw[thick] (28) -- (12);
\draw[thick] (10) -- (27);
    \draw[thick] (11) -- (27);
    
     \draw[thick] (29) -- (6);
\draw[thick] (29) -- (7);
    \draw[thick] (29) -- (30); 
    \draw[thick] (8) -- (30);
    
     \draw[thick] (9) -- (30);
\draw[thick] (2) -- (30);
    \draw[thick] (3) -- (30); 
    \draw[thick] (4) -- (30);
    \draw[thick] (5) -- (30);
    \draw[thick] (27) -- (28);

    \draw[thick] (31) -- (20);
    \draw[thick] (31) -- (21);

    \draw[thick] (31) -- (16);
    \draw[thick] (31) -- (25);
    \draw[thick] (31) -- (32);
    \draw[thick] (31) -- (17);

    \draw[thick] (32) -- (18);
    \draw[thick] (32) -- (24);
    \draw[thick] (32) -- (19);
    \draw[thick] (32) -- (31);

    \draw[thick] (33) -- (12);
    \draw[thick] (33) -- (27);
    \draw[thick] (33) -- (34);
    \draw[thick] (33) -- (13);

    \draw[thick] (34) -- (14);
    \draw[thick] (34) -- (26);
    \draw[thick] (34) -- (33);
    \draw[thick] (34) -- (15);
    \draw[thick] (34) -- (10);
    \draw[thick] (34) -- (28);
    \draw[thick] (34) -- (35);
    \draw[thick] (34) -- (11);

    \draw[thick] (35) -- (12);
    \draw[thick] (35) -- (27);
    \draw[thick] (35) -- (34);
    \draw[thick] (35) -- (13);

    \draw[thick] (36) -- (6);
    \draw[thick] (36) -- (30);
    \draw[thick] (36) -- (37);
    \draw[thick] (36) -- (7);
     
     \draw[thick] (37) -- (8);
    \draw[thick] (37) -- (29);
    \draw[thick] (37) -- (36);
    \draw[thick] (37) -- (9);
      
      \draw[thick] (37) -- (5);
    \draw[thick] (37) -- (4);
    \draw[thick] (37) -- (3);
    \draw[thick] (37) -- (2);

\end{scope}

\begin{scope}[xshift=12.7cm] 
    
   \node (1)[label={[font=\small]below:\text{($s$+2,0)}}]at (0.85,1) {};
    \node (2)[label={[font=\scriptsize]left:($s$+1,0)}] at (-2.8,3) {};
    \node (3) at (-0.85,3) {};
    \node (4) at (2.55,3) {};
    \node (5)[label={[font=\scriptsize]right:($s$+1,$2^{s+3}$)}] at (4.5,3) {};
    
    \node (6)[label={[font=\small]left:($s$,0)}] at (-1.7,5) {};
    \node (7)[label={[font=\small]right:($s$,3)}] at (3.4,5) {};
    \node (8)[label={[font=\small]left:($s$-1,0)}] at (-1.7,6.5) {};
    \node (9)[label={[font=\small]right:($s$-1,3)}] at (3.4,6.5) {};
   
    \node (10)[label={[font=\small]left:($m$+1,0)}] at (-1.7,9) {};
    \node (11)[label={[font=\small]right:($m$+1,3)}] at (3.4,9) {};
    \node (14)[label={[font=\small]left:($m$-1,0)}] at (-1.7,12) {};
    \node (15)[label={[font=\small]right:($m$-1,3)}] at (3.4,12) {};

   \node (16)[label={[font=\small]left:(2,0)}] at (-1.7,14.5) {};
   \node (17)[label={[font=\small]right:(2,3)}] at (3.4,14.5) {};
   \node (18)[label={[font=\small]left:(1,0)}] at (-1.7,16) {};
   \node (19)[label={[font=\small]right:(1,3)}] at (3.4,16) {};
    
   \node (20)[label={[font=\small]left:(0,0)}] at (0,17.5) {};
   \node (21)[label={[font=\small]right:(0,1)}] at (1.7,17.5) {};

 \node (24) at (0,16) {};
 \node (25) at (0,14.5) {};
  \node (26) at (0,12) {};
  \node (28) at (0,9) {};
  \node (29) at (0,6.5) {};
  \node (30) at (0,5) {};

 \node (31) at (1.7,16) {};
 \node (32) at (1.7,14.5) {};
  \node (33) at (1.7,12) {};
  \node (35) at (1.7,9) {};
 \node (36) at (1.7,6.5) {};
  \node (37) at (1.7,5) {};

  \node (40)[label={[font=\small,label distance=0.8cm]180:($m$,0)}] at (0.85,10.5) {};

\draw[thin, dashed, dash pattern=on 1.6pt off 1pt] (8) -- ++(0,0.6); 
\draw[thin, dashed, dash pattern=on 1.6pt off 1pt] (9) -- ++(0,0.6); 
\draw[thin, dashed, dash pattern=on 1.6pt off 1pt] (10) -- ++(0,-0.6); 
\draw[thin, dashed, dash pattern=on 1.6pt off 1pt] (11) -- ++(0,-0.6);     

\draw[thin, dashed, dash pattern=on 1.6pt off 1pt] (29) -- ++(0,0.6); 
\draw[thin, dashed, dash pattern=on 1.6pt off 1pt] (28) -- ++(0,-0.6); 

\draw[thin, dashed, dash pattern=on 1.6pt off 1pt] (35) -- ++(0,-0.6); 
\draw[thin, dashed, dash pattern=on 1.6pt off 1pt] (36) -- ++(0,0.6);

\draw[thin, dashed, dash pattern=on 1.6pt off 1pt] (14) -- ++(0,0.6); 
\draw[thin, dashed, dash pattern=on 1.6pt off 1pt] (15) -- ++(0,0.6); 
\draw[thin, dashed, dash pattern=on 1.6pt off 1pt] (16) -- ++(0,-0.6); 
\draw[thin, dashed, dash pattern=on 1.6pt off 1pt] (17) -- ++(0,-0.6);     

\draw[thin, dashed, dash pattern=on 1.6pt off 1pt] (26) -- ++(0,0.6); 
\draw[thin, dashed, dash pattern=on 1.6pt off 1pt] (25) -- ++(0,-0.6); 
\draw[thin, dashed, dash pattern=on 1.6pt off 1pt] (32) -- ++(0,-0.6); 
\draw[thin, dashed, dash pattern=on 1.6pt off 1pt] (33) -- ++(0,0.6);     
       
  \fill (-0.05,3) circle (1.5pt);
   \fill (0.2,3) circle (1.5pt);
   \fill (0.45,3) circle (1.5pt);
   
   \fill (1.75,3) circle (1.5pt);
   \fill (1.5,3) circle (1.5pt);
   \fill (1.25,3) circle (1.5pt);

    \draw[thick] (1) -- (2);
    \draw[thick] (1) -- (3);
    \draw[thick] (1) -- (4);
    \draw[thick] (1) -- (5);
    
    \draw[thick] (2) -- (6);
    \draw[thick] (2) -- (7);
     \draw[thick] (3) -- (6);
    \draw[thick] (3) -- (7);
     \draw[thick] (4) -- (6);
    \draw[thick] (4) -- (7);
     \draw[thick] (5) -- (6);
    \draw[thick] (5) -- (7); 
    
    \draw[thick] (6) -- (8);
    \draw[thick] (6) -- (9);
    \draw[thick] (7) -- (8);
    \draw[thick] (7) -- (9);

    \draw[thick] (16) -- (18);
    \draw[thick] (16) -- (19);
    \draw[thick] (17) -- (18);
    \draw[thick] (17) -- (19);

\draw[thick] (18) -- (20);
    \draw[thick] (18) -- (21);
    \draw[thick] (19) -- (20);
    \draw[thick] (19) -- (21);

\draw[thick] (20) -- (24);
    \draw[thick] (21) -- (24);
    \draw[thick] (16) -- (24);
    \draw[thick] (25) -- (24);
\draw[thick] (17) -- (24);
    \draw[thick] (25) -- (18);
    \draw[thick] (25) -- (19);

     \draw[thick] (29) -- (6);
\draw[thick] (29) -- (7);
    \draw[thick] (29) -- (30); 
    \draw[thick] (8) -- (30);
    
     \draw[thick] (9) -- (30);
\draw[thick] (2) -- (30);
    \draw[thick] (3) -- (30); 
    \draw[thick] (4) -- (30);
    \draw[thick] (5) -- (30);

    \draw[thick] (31) -- (20);
    \draw[thick] (31) -- (21);

    \draw[thick] (31) -- (16);
    \draw[thick] (31) -- (25);
    \draw[thick] (31) -- (32);
    \draw[thick] (31) -- (17);

    \draw[thick] (32) -- (18);
    \draw[thick] (32) -- (24);
    \draw[thick] (32) -- (19);
    \draw[thick] (32) -- (31);

     \draw[thick] (36) -- (6);
    \draw[thick] (36) -- (30);
    \draw[thick] (36) -- (37);
    \draw[thick] (36) -- (7);
     
     \draw[thick] (37) -- (8);
    \draw[thick] (37) -- (29);
    \draw[thick] (37) -- (36);
    \draw[thick] (37) -- (9);
      
      \draw[thick] (37) -- (5);
    \draw[thick] (37) -- (4);
    \draw[thick] (37) -- (3);
    \draw[thick] (37) -- (2);

    \draw[thick] (40) -- (14);
        \draw[thick] (40) -- (15);
    \draw[thick] (40) -- (26);
    \draw[thick] (40) -- (33);
    
    \draw[thick] (40) -- (10);
    \draw[thick] (40) -- (11);
    \draw[thick] (40) -- (28);
    \draw[thick] (40) -- (35);

\end{scope}

\end{tikzpicture}

\caption{The frames $\mathbf{T}_{s}$ and $\mathbf{T}'_{s,m}$}
\label{fig:shatrov}
\end{figure}
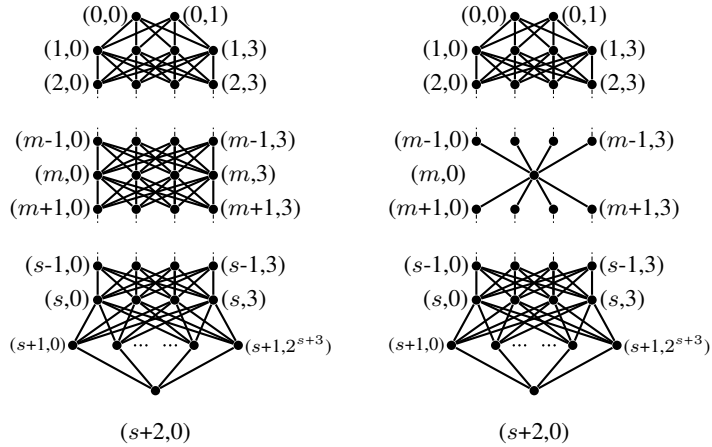

In \cite[\S~5]{Fon07}, Fontaine gave two kinds of structures $\mathbf{G}_s$ and $\mathbf{G}'_{s,m}$ as frames depicted in Figure \ref{fig:gaelle}, she proved that $\mathbf{G}_{s}$ is not a $\mathsf{Cheq}$-frame for any $s\geq 1$ and want to prove that $\mathbf{G}'_{s,m}$ is a $\mathsf{Cheq}$-frame. 

\begin{lemma}
	Suppose $\mathbf{F}_0=(F_0, \leq_0)$ be a frame with a top, and $\mathbf{F}_1=(F_1, \leq_1)$ a frame with a root. If $\mathbf{F}_0$ is a p-morphic image of $\mathbf{V}_n$ and $\mathbf{F}_1$ is a p-morphic image of $\mathbf{V}_m$, it follows that $\mathbf{F}_0\overline\oplus\mathbf{F}_1$ is a p-morphic image of $\mathbf{V}_{n+m}$.
\label{lem:verticalsumcheq}	
\end{lemma}

Lemma \ref{lem:verticalsumcheq} gives us a motivation of focusing on the $(m,0){\uparrow}$ in $\mathbf{G}'_{s,m}$ since $(m,0){\downarrow}$ is already a $\mathsf{Med}-frame$ and thus a $\mathsf{Cheq}$-frame. Let $\mathbf{Fon}_m$ be the frame $(m,0){\uparrow}$, then Fontaine provided the following proposition.
\begin{proposition}
	If $\mathbf{Fon}_m$ is $\mathsf{Cheq}$-frame for any $m$, then $\mathsf{Cheq}$ is not finitely axiomatizable.
\end{proposition}


 

\begin{definition}
A collection $\{X_0,\ldots, X_{n-1}\}$ is called  a \emph{full} $n$-\emph{partition} of $X$ with respect to $Y$ if:
\begin{enumerate}
	\item $X=\bigcup_{i=0}^{n-1}X_i,$
	\item $X_i\cap X_j=\emptyset$ for $0\leq i\neq j \leq n-1,$
	\item for each $i$ and for every $y\in Y$, there exists $x_i \in X_i$ satisfying $y\leq x_i$.
\end{enumerate} 	
\end{definition}

Consider the set $D(i,j)$, consisting of all $x$ in the frame $\mathbf{V}_i$, where the depth $d(x)$ equals $j+1$. Fontaine showed that if, for every $i > 1$, a 3-full partition of $D(i, 1)$ with respect to $D(i, 2)$ exists, then it is possible to construct a 3-full partition of $D(i, j)$ with respect to $D(i, j+1)$ for any $j > 1$. Consequently, Fontaine suggests the following proposition in \cite[Proposition 22]{Fon07}.

\begin{proposition}
If for any $i > 1$, a 3-full partition of $D(i, 1)$ with respect to $D(i, 2)$ exists, then $\mathbf{Fon}_m$ is a $\mathsf{Cheq}$-frame for any $m$ and so $\mathsf{Cheq}$ is not finitely axiomatizable.
\label{pro:fontainconjecture}	
\end{proposition}

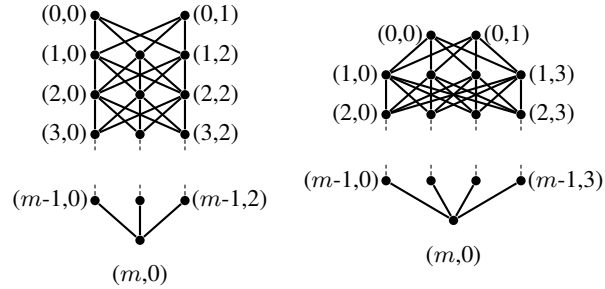
\begin{figure}[H]
\centering
\begin{tikzpicture}[scale=0.35, node distance=2cm, every node/.style={circle, fill=black, minimum size=3.5pt, inner sep=0pt}]

\begin{scope}[xshift=0cm] 
       
    \node (14)[label={[font=\small]left:($m$-1,0)}] at (-1.7,12) {};
    \node (15)[label={[font=\small]right:($m$-1,2)}] at (1.7,12) {};
   
   \node (16)[label={[font=\small]left:(3,0)}] at (-1.7,14.5) {};
   \node (17)[label={[font=\small]right:(3,2)}] at (1.7,14.5) {};
   \node (18)[label={[font=\small]left:(2,0)}] at (-1.7,16) {};
   \node (19)[label={[font=\small]right:(2,2)}] at (1.7,16) {};
    
   \node (20)[label={[font=\small]left:(1,0)}] at (-1.7,17.5) {};
   \node (21)[label={[font=\small]right:(1,2)}] at (1.7,17.5) {};
   \node (22)[label={[font=\small]left:(0,0)}] at (-1.7,19) {};
   \node (23)[label={[font=\small]right:(0,1)}] at (1.7,19) {};
  
   \node (31) at (0,17.5) {};

 \node (24) at (0,16) {};
 \node (25) at (0,14.5) {};
  \node (26) at (0,12) {};

 \node (27)[label={[font=\small]below:($m$,0)}] at (0,10.5) {};

\draw[thin, dashed, dash pattern=on 1.6pt off 1pt] (14) -- ++(0,0.6); 
\draw[thin, dashed, dash pattern=on 1.6pt off 1pt] (15) -- ++(0,0.6); 
\draw[thin, dashed, dash pattern=on 1.6pt off 1pt] (16) -- ++(0,-0.6); 
\draw[thin, dashed, dash pattern=on 1.6pt off 1pt] (17) -- ++(0,-0.6);     

\draw[thin, dashed, dash pattern=on 1.6pt off 1pt] (26) -- ++(0,0.6); 
\draw[thin, dashed, dash pattern=on 1.6pt off 1pt] (25) -- ++(0,-0.6);

    \draw[thick] (16) -- (18);
    \draw[thick] (16) -- (19);
    \draw[thick] (17) -- (18);
    \draw[thick] (17) -- (19);

\draw[thick] (18) -- (20);
    \draw[thick] (18) -- (21);
    \draw[thick] (19) -- (20);
    \draw[thick] (19) -- (21);
\draw[thick] (20) -- (22);
    \draw[thick] (20) -- (23);
    \draw[thick] (21) -- (22);
    \draw[thick] (21) -- (23);

\draw[thick] (20) -- (24);
    \draw[thick] (21) -- (24);
    \draw[thick] (16) -- (24);
    \draw[thick] (25) -- (24);
\draw[thick] (17) -- (24);
    \draw[thick] (25) -- (18);
    \draw[thick] (25) -- (19);

\draw[thick] (26) -- (27);
 
 \draw[thick] (27) -- (14);
\draw[thick] (15) -- (27);

 \draw[thick] (31) -- (23);
 \draw[thick] (31) -- (22);
 \draw[thick] (31) -- (18);
  \draw[thick] (31) -- (19);
  \draw[thick] (31) -- (24);

\end{scope}

\begin{scope}[xshift=11cm,yshift=0.75cm]

    \node (14)[label={[font=\small]left:($m$-1,0)}] at (-1.7,12) {};
    \node (15)[label={[font=\small]right:($m$-1,3)}] at (3.4,12) {};

   \node (16)[label={[font=\small]left:(2,0)}] at (-1.7,14.5) {};
   \node (17)[label={[font=\small]right:(2,3)}] at (3.4,14.5) {};
   \node (18)[label={[font=\small]left:(1,0)}] at (-1.7,16) {};
   \node (19)[label={[font=\small]right:(1,3)}] at (3.4,16) {};
    
   \node (20)[label={[font=\small]left:(0,0)}] at (0,17.5) {};
   \node (21)[label={[font=\small]right:(0,1)}] at (1.7,17.5) {};

 \node (24) at (0,16) {};
 \node (25) at (0,14.5) {};
  \node (26) at (0,12) {};

 \node (31) at (1.7,16) {};
 \node (32) at (1.7,14.5) {};
  \node (33) at (1.7,12) {};

  \node (40)[label={[font=\small]below:($m$,0)}] at (0.85,10.5) {};

\draw[thin, dashed, dash pattern=on 1.6pt off 1pt] (14) -- ++(0,0.6); 
\draw[thin, dashed, dash pattern=on 1.6pt off 1pt] (15) -- ++(0,0.6); 
\draw[thin, dashed, dash pattern=on 1.6pt off 1pt] (16) -- ++(0,-0.6); 
\draw[thin, dashed, dash pattern=on 1.6pt off 1pt] (17) -- ++(0,-0.6);     

\draw[thin, dashed, dash pattern=on 1.6pt off 1pt] (26) -- ++(0,0.6); 
\draw[thin, dashed, dash pattern=on 1.6pt off 1pt] (25) -- ++(0,-0.6); 
\draw[thin, dashed, dash pattern=on 1.6pt off 1pt] (32) -- ++(0,-0.6); 
\draw[thin, dashed, dash pattern=on 1.6pt off 1pt] (33) -- ++(0,0.6);     
             
    \draw[thick] (16) -- (18);
    \draw[thick] (16) -- (19);
    \draw[thick] (17) -- (18);
    \draw[thick] (17) -- (19);

\draw[thick] (18) -- (20);
    \draw[thick] (18) -- (21);
    \draw[thick] (19) -- (20);
    \draw[thick] (19) -- (21);

\draw[thick] (20) -- (24);
    \draw[thick] (21) -- (24);
    \draw[thick] (16) -- (24);
    \draw[thick] (25) -- (24);
\draw[thick] (17) -- (24);
    \draw[thick] (25) -- (18);
    \draw[thick] (25) -- (19);

    \draw[thick] (31) -- (20);
    \draw[thick] (31) -- (21);

    \draw[thick] (31) -- (16);
    \draw[thick] (31) -- (25);
    \draw[thick] (31) -- (32);
    \draw[thick] (31) -- (17);

    \draw[thick] (32) -- (18);
    \draw[thick] (32) -- (24);
    \draw[thick] (32) -- (19);
    \draw[thick] (32) -- (31);
     
    \draw[thick] (40) -- (14);
        \draw[thick] (40) -- (15);
    \draw[thick] (40) -- (26);
    \draw[thick] (40) -- (33);

\end{scope}

\end{tikzpicture}
\caption{The frames $\mathbf{Fon}_m$ and $\mathbf{Sha}_m$}
\end{figure}
According to Theorem \ref{thm:dualfaceposetofcube}, each $\mathbf{V}_i$ can be regarded as the dual face poset of hypercube $Q_i$. Then $D(i,1)$ is the set of edges of $Q_i$, and $D(i,2)$ is the set of faces of $Q_i$. The assumption of Proposition \ref{pro:fontainconjecture} is to ask to have edge-coloring $\mu$ of any hypercube $Q_i$ with $3$ colors and $\mu$ is $Q_2$-polychromatic 3-coloring of $Q_i$ since every face of $Q_i$ contains every color. But by Theorem \ref{thm:offner}, the polychromatic number of $Q_2$ 
\[
p(Q_2)=\left\lfloor \frac{(2+1)^2}{4} \right\rfloor=2
\]  
Thus for a sufficiently large integer $N$, the hypercube $Q_N$ cannot allow a $Q_2$-polychromatic 3-coloring, then the assumption of Proposition \ref{pro:fontainconjecture} fails.

Shatrov claimed to prove the non-finitely axiomatizability of $\mathsf{Cheq}$ by showing that $\mathbf{T}'_{s,m}$ is a $\mathsf{Cheq}$-frame while $\mathbf{T}_s$ is not, where $\mathbf{T}_s$ and $\mathbf{T}'_{s,m}$ are frames depicted in Figure \ref{fig:shatrov}. His strategy involves proving that each frame $\mathbf{Sha}_m$, which is $(m,0){\uparrow}$ in $\mathbf{T}'_{s,m}$, is a $\mathsf{Cheq}$-frame. If we consider the possibility of a 4-full partition of $D(i,1)$ with respect to $D(i,2)$ to achieve this goal, it fails due to $p(Q_2)=2$. However, since Shatrov has not published his proof, we do not know his concrete assumptions for proving that $\mathbf{Sha}_m \vDash \mathsf{Cheq}$.

\subsection{Kuznetsov's alternative}
After briefly describing the infeasibility of the above frames to prove that $\mathsf{Cheq}$ is not finitely axiomatizable, Kuznetsov provided his strategy by building the following alternative structure in \cite{Kuz19}.
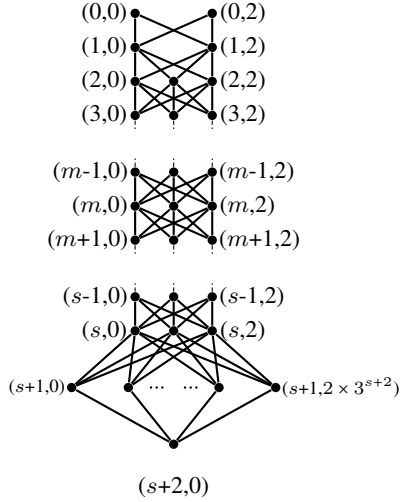
\begin{figure}[H]
\centering
\begin{tikzpicture}[scale=0.3, node distance=2cm, every node/.style={circle, fill=black, minimum size=3.5pt, inner sep=0pt}]

\begin{scope}[xshift=0cm] 
    
    \node (1)[label={[font=\small]below:\text{($s$+2,0)}}]at (0,0) {};
    \node (2)[label={[font=\scriptsize]left:($s$+1,0)}] at (-4.5,2.5) {};
    \node (3) at (-2,2.5) {};
    \node (4) at (2,2.5) {};
    \node (5)[label={[font=\scriptsize]right:($s$+1,$2\times3^{s+2}$)}] at (4.5,2.5) {};
    \node (6)[label={[font=\small]left:($s$,0)}] at (-1.7,5) {};
    \node (7)[label={[font=\small]right:($s$,2)}] at (1.7,5) {};
    \node (8)[label={[font=\small]left:($s$-1,0)}] at (-1.7,6.5) {};
    \node (9)[label={[font=\small]right:($s$-1,2)}] at (1.7,6.5) {};
    \node (10)[label={[font=\small]left:($m$+1,0)}] at (-1.7,9) {};
    \node (11)[label={[font=\small]right:($m$+1,2)}] at (1.7,9) {};
    \node (12)[label={[font=\small]left:($m$,0)}] at (-1.7,10.5) {};
    \node (13)[label={[font=\small]right:($m$,2)}] at (1.7,10.5) {};
    \node (14)[label={[font=\small]left:($m$-1,0)}] at (-1.7,12) {};
    \node (15)[label={[font=\small]right:($m$-1,2)}] at (1.7,12) {};

   \node (16)[label={[font=\small]left:(3,0)}] at (-1.7,14.5) {};
   \node (17)[label={[font=\small]right:(3,2)}] at (1.7,14.5) {};
   \node (18)[label={[font=\small]left:(2,0)}] at (-1.7,16) {};
   \node (19)[label={[font=\small]right:(2,2)}] at (1.7,16) {};
    
   \node (20)[label={[font=\small]left:(1,0)}] at (-1.7,17.5) {};
   \node (21)[label={[font=\small]right:(1,2)}] at (1.7,17.5) {};
   \node (22)[label={[font=\small]left:(0,0)}] at (-1.7,19) {};
   \node (23)[label={[font=\small]right:(0,2)}] at (1.7,19) {};

 \node (24) at (0,16) {};
 \node (25) at (0,14.5) {};
  \node (26) at (0,12) {};
  \node (27) at (0,10.5) {};
  \node (28) at (0,9) {};
  \node (29) at (0,6.5) {};
  \node (30) at (0,5) {};

\draw[thin, dashed, dash pattern=on 1.6pt off 1pt] (8) -- ++(0,0.6); 
\draw[thin, dashed, dash pattern=on 1.6pt off 1pt] (9) -- ++(0,0.6); 
\draw[thin, dashed, dash pattern=on 1.6pt off 1pt] (10) -- ++(0,-0.6); 
\draw[thin, dashed, dash pattern=on 1.6pt off 1pt] (11) -- ++(0,-0.6);     

\draw[thin, dashed, dash pattern=on 1.6pt off 1pt] (29) -- ++(0,0.6); 
\draw[thin, dashed, dash pattern=on 1.6pt off 1pt] (28) -- ++(0,-0.6);

\draw[thin, dashed, dash pattern=on 1.6pt off 1pt] (14) -- ++(0,0.6); 
\draw[thin, dashed, dash pattern=on 1.6pt off 1pt] (15) -- ++(0,0.6); 
\draw[thin, dashed, dash pattern=on 1.6pt off 1pt] (16) -- ++(0,-0.6); 
\draw[thin, dashed, dash pattern=on 1.6pt off 1pt] (17) -- ++(0,-0.6);     

\draw[thin, dashed, dash pattern=on 1.6pt off 1pt] (26) -- ++(0,0.6); 
\draw[thin, dashed, dash pattern=on 1.6pt off 1pt] (25) -- ++(0,-0.6);

   \fill (-0.5,2.5) circle (1.5pt);
    \fill (-1,2.5) circle (1.5pt);
    \fill (-0.75,2.5) circle (1.5pt);
  \fill (1,2.5) circle (1.5pt);
    \fill (0.75,2.5) circle (1.5pt);
    \fill (0.5,2.5) circle (1.5pt);

    \draw[thick] (1) -- (2);
    \draw[thick] (1) -- (3);
    \draw[thick] (1) -- (4);
    \draw[thick] (1) -- (5);
    
    \draw[thick] (2) -- (6);
    \draw[thick] (2) -- (7);
     \draw[thick] (3) -- (6);
    \draw[thick] (3) -- (7);
     \draw[thick] (4) -- (6);
    \draw[thick] (4) -- (7);
     \draw[thick] (5) -- (6);
    \draw[thick] (5) -- (7); 
    
    \draw[thick] (6) -- (8);
    \draw[thick] (6) -- (9);
    \draw[thick] (7) -- (8);
    \draw[thick] (7) -- (9);

    \draw[thick] (10) -- (12);
    \draw[thick] (10) -- (13);
    \draw[thick] (11) -- (12);
    \draw[thick] (11) -- (13);
    \draw[thick] (12) -- (14);
    \draw[thick] (12) -- (15);
    \draw[thick] (13) -- (14);
    \draw[thick] (13) -- (15);
   
    \draw[thick] (16) -- (18);
    \draw[thick] (16) -- (19);
    \draw[thick] (17) -- (18);
    \draw[thick] (17) -- (19);

\draw[thick] (18) -- (20);
    \draw[thick] (18) -- (21);
    \draw[thick] (19) -- (20);
    \draw[thick] (19) -- (21);
\draw[thick] (20) -- (22);
    \draw[thick] (20) -- (23);
    \draw[thick] (21) -- (22);
    \draw[thick] (21) -- (23);

\draw[thick] (20) -- (24);
    \draw[thick] (21) -- (24);
    \draw[thick] (16) -- (24);
    \draw[thick] (25) -- (24);
\draw[thick] (17) -- (24);
    \draw[thick] (25) -- (18);
    \draw[thick] (25) -- (19);

    \draw[thick] (26) -- (12);
\draw[thick] (26) -- (27);
    \draw[thick] (26) -- (13);
 
 \draw[thick] (27) -- (14);
\draw[thick] (15) -- (27);
    \draw[thick] (28) -- (13); 
    \draw[thick] (28) -- (12);
\draw[thick] (10) -- (27);
    \draw[thick] (11) -- (27);
    
     \draw[thick] (29) -- (6);
\draw[thick] (29) -- (7);
    \draw[thick] (29) -- (30); 
    \draw[thick] (8) -- (30);
    
    \draw[thick] (9) -- (30);
    \draw[thick] (2) -- (30);
    \draw[thick] (3) -- (30); 
    \draw[thick] (4) -- (30);
    \draw[thick] (5) -- (30);
    \draw[thick] (27) -- (28);

\end{scope}

\end{tikzpicture}
\caption{The frame $\mathbf{E}_s$}
\label{fig:kuz}
\end{figure}

We define the \emph{Kuznetsov frame} $\mathbf{E}_s$ to be the frame shown in Figure \ref{fig:kuz}.
The result of the polychromatic number of $Q_2$ cannot directly negate the Kuznetsov frame $\mathbf{E}_s$. Kuznetsov aimed to prove that each Kuznetsov frame $\mathbf{E}_s$ is not a $\mathsf{Cheq}$-frame as the key step to follow the established method of proving that $\mathsf{Cheq}$ is not finitely axiomatizable. But we have the following proposition:

\begin{proposition}
	$\mathbf{E}_s$ is a $\mathsf{Cheq}$-frame. 
\end{proposition}
\begin{proof}
	Each $\mathbf{E}_s$ is a linear sum of a finite rooted frame and the Bowtie frame. By the Bowtie Lemma~\ref{lem:bow-tie}, every Kuznetsov frame $\mathbf{E}_s$ is a $\mathsf{Cheq}$-frame. 
\end{proof}

\section{$\mathsf{Cheq}$ is not finitely axiomatizable}
To prove the main result of this paper, we first settle the open problem for $\mathsf{Cheq}$. By standard results concerning modal companions, the result also settles the case of $\mathsf{L}_{\infty}$. The combinatorial structure of the frame $\mathbf{V}_n$, together with Bowtie lemma, suggest the structures used in our construction. For any $x\in \mathbf{V}_n$, the parities of the numbers of $w_1$ to the left of the leftmost $w_0$ and to the right of the rightmost $w_0$, counted respectively, form a coloring strategy.

\begin{definition}
	Given integers $u\geq 3$ and $n\geq 1$, let $\mathbf{C}_{u,v}$ denote the frame shown in Figure~\ref{fig:china}. There are $n$ red points, each of depth $u$. 
\end{definition}
The immediate successors of the root are $\{x,y\}$ and red points. The immediate successors of $x$ are $\{(u-2,0),(u-2,1),(u-2,3)\}$, those of $y$ are $\{(u-2,0),(u-2,2),(u-2,3)\}$, and those of each $z^i$ are $\{(u-2,1), (u-2,2), (u-2,3)\}$. For any point $a=(i,j)$, $2\leq i \leq u-2$ and $0\leq j \leq3$, it cannot see $(i-1, 3-j)$.
	
\begin{figure}[H]
\centering
\begin{tikzpicture}[scale=0.5, node distance=2cm, every node/.style={circle, fill=black, minimum size=3.5pt, inner sep=0pt}]

    \node (1)[label={[label distance=0.1cm]below:{\normalsize $(u,0)$}}] at (0.85,0) {};
    \node (2)[label={[label distance=0.1cm]left:{\normalsize $x$}}] at (-3.8,2.5) {};
    \node (3)[label={[label distance=0.1cm]left:{\normalsize $y$}}]  at (-2.1,2.5) {};
    \node (4) [label={[label distance=0.1cm]right:{\normalsize $z^1$}}]  at (-0.4,2.5)[fill=red, circle] {};
    \node (55) [label={[label distance=0.1cm]right:{\normalsize $z^2$}}]  at (1.3,2.5)[fill=red, circle] {};
    
    \node (5)[label={[label distance=0.1cm]right:{\normalsize $z^n$}}] [fill=red, circle] at (5.5,2.5) {};
    \node (6) [label={[label distance=0.1cm]left:{\normalsize $(u-2,0)$}}]  at (-1.7,5) {};
    \node (7) [label={[label distance=0.1cm]right:{\normalsize $(u-2,3)$}}] at (3.4,5) {};
   
    \node (8)[label={[label distance=0.1cm]left:{\normalsize $(u-3,0)$}}] at (-1.7,6.5) {};
    \node (9)[label={[label distance=0.1cm]right:{\normalsize $(u-3,3)$}}] at (3.4,6.5) {};
   
    \node (10)[label={[label distance=0.1cm]left:{\normalsize $(m+1,0)$}}] at (-1.7,9) {};
    \node (11)[label={[label distance=0.1cm]right:{\normalsize $(m+1,3)$}}] at (3.4,9) {};
    \node (12)[label={[label distance=0.1cm]left:{\normalsize $(m,0)$}}] at (-1.7,10.5) {};
    \node (13)[label={[label distance=0.1cm]right:{\normalsize $(m,3)$}}]at (3.4,10.5) {};
    \node (14)[label={[label distance=0.1cm]left:{\normalsize $(m-1,0)$}}] at (-1.7,12) {};
    \node (15)[label={[label distance=0.1cm]right:{\normalsize $(m-1,3)$}}] at (3.4,12) {};

   \node (16)[label={[label distance=0.1cm]left:{\normalsize $(2,0)$}}] at (-1.7,14.5) {};
   \node (17)[label={[label distance=0.1cm]right:{\normalsize $(2,3)$}}] at (3.4,14.5) {};
   \node (18)[label={[label distance=0.1cm]left:{\normalsize $(1,0)$}}] at (-1.7,16) {};
   \node (19)[label={[label distance=0.1cm]right:{\normalsize $(1,3)$}}] at (3.4,16) {};
    
   \node (20)[label={[label distance=0.1cm]left:{\normalsize $(0,0)$}}] at (-1.7,17.5) {};
   \node (21)[label={[label distance=0.1cm]right:{\normalsize $(0,1)$}}] at (3.4,17.5) {};

 \node (24) at (0,16) {};
 \node (25) at (0,14.5) {};
  \node (26) at (0,12) {};
  \node (27) at (0,10.5) {};
  \node (28) at (0,9) {};
  \node (29) at (0,6.5) {};
  \node (30) at (0,5) {};

 \node (31) at (1.7,16) {};
 \node (32) at (1.7,14.5) {};
  \node (33) at (1.7,12) {};
  \node (34) at (1.7,10.5) {};
  \node (35) at (1.7,9) {};
 \node (36) at (1.7,6.5) {};
 \node (37) at (1.7,5) {};

\draw[thin, dashed, dash pattern=on 1.6pt off 1pt] (8) -- ++(0,0.6); 
\draw[thin, dashed, dash pattern=on 1.6pt off 1pt] (9) -- ++(0,0.6); 
\draw[thin, dashed, dash pattern=on 1.6pt off 1pt] (10) -- ++(0,-0.6); 
\draw[thin, dashed, dash pattern=on 1.6pt off 1pt] (11) -- ++(0,-0.6);     

\draw[thin, dashed, dash pattern=on 1.6pt off 1pt] (29) -- ++(0,0.6); 
\draw[thin, dashed, dash pattern=on 1.6pt off 1pt] (28) -- ++(0,-0.6); 

\draw[thin, dashed, dash pattern=on 1.6pt off 1pt] (35) -- ++(0,-0.6); 
\draw[thin, dashed, dash pattern=on 1.6pt off 1pt] (36) -- ++(0,0.6);

\draw[thin, dashed, dash pattern=on 1.6pt off 1pt] (14) -- ++(0,0.6); 
\draw[thin, dashed, dash pattern=on 1.6pt off 1pt] (15) -- ++(0,0.6); 
\draw[thin, dashed, dash pattern=on 1.6pt off 1pt] (16) -- ++(0,-0.6); 
\draw[thin, dashed, dash pattern=on 1.6pt off 1pt] (17) -- ++(0,-0.6);     

\draw[thin, dashed, dash pattern=on 1.6pt off 1pt] (26) -- ++(0,0.6); 
\draw[thin, dashed, dash pattern=on 1.6pt off 1pt] (25) -- ++(0,-0.6); 
\draw[thin, dashed, dash pattern=on 1.6pt off 1pt] (32) -- ++(0,-0.6); 
\draw[thin, dashed, dash pattern=on 1.6pt off 1pt] (33) -- ++(0,0.6);     
       
   \fill[red] (2.5,2.5) circle (1.5pt);
   \fill[red] (2.75,2.5) circle (1.5pt);
   \fill[red] (3,2.5) circle (1.5pt);
   
   \fill[red] (4.3,2.5) circle (1.5pt);
   \fill[red] (4.05,2.5) circle (1.5pt);
   \fill[red] (3.8,2.5) circle (1.5pt);

   \draw[thick] (1) -- (2);
    \draw[thick] (1) -- (3);
    \draw[thick] (1) -- (4);
    \draw[thick] (1) -- (5);
    
    \draw[thick] (2) -- (6);
    \draw[thick] (2) -- (7);
     \draw[thick] (3) -- (6);
    \draw[thick] (3) -- (7);
    \draw[thick] (4) -- (7);
    \draw[thick] (5) -- (7); 
    
    \draw[thick] (6) -- (8);
  \draw[thick] (7) -- (9);

    \draw[thick] (10) -- (12);
   \draw[thick] (11) -- (13);
    \draw[thick] (12) -- (14);
   \draw[thick] (13) -- (15);
   
   \draw[thick] (16) -- (18);
   \draw[thick] (17) -- (19);

\draw[thick] (18) -- (20);
    \draw[thick] (18) -- (21);
    \draw[thick] (19) -- (20);
    \draw[thick] (19) -- (21);

\draw[thick] (20) -- (24);
    \draw[thick] (21) -- (24);
    \draw[thick] (16) -- (24);
    \draw[thick] (25) -- (24);
\draw[thick] (17) -- (24);
    \draw[thick] (25) -- (18);
    \draw[thick] (25) -- (19);

    \draw[thick] (26) -- (12);
\draw[thick] (26) -- (27);
    \draw[thick] (26) -- (13);
 
 \draw[thick] (27) -- (14);
\draw[thick] (15) -- (27);
    \draw[thick] (28) -- (13); 
    \draw[thick] (28) -- (12);
\draw[thick] (10) -- (27);
    \draw[thick] (11) -- (27);
    
     \draw[thick] (29) -- (6);
\draw[thick] (29) -- (7);
    \draw[thick] (29) -- (30); 
    \draw[thick] (8) -- (30);
    
     \draw[thick] (9) -- (30);
\draw[thick] (2) -- (30);
    \draw[thick] (4) -- (30);
    \draw[thick] (5) -- (30);
 \draw[thick] (27) -- (28);

    \draw[thick] (31) -- (20);
    \draw[thick] (31) -- (21);

    \draw[thick] (31) -- (16);
    \draw[thick] (31) -- (32);
    \draw[thick] (31) -- (17);

    \draw[thick] (32) -- (18);
    \draw[thick] (32) -- (19);

    \draw[thick] (33) -- (12);
    \draw[thick] (33) -- (34);
    \draw[thick] (33) -- (13);

    \draw[thick] (34) -- (14);
    \draw[thick] (34) -- (15);
    \draw[thick] (34) -- (10);
    \draw[thick] (34) -- (35);
    \draw[thick] (34) -- (11);

    \draw[thick] (35) -- (12);
    \draw[thick] (35) -- (13);

    \draw[thick] (36) -- (6);
\draw[thick] (36) -- (37);
    \draw[thick] (36) -- (7);
     
     \draw[thick] (37) -- (8);
    \draw[thick] (37) -- (9);
      
      \draw[thick] (37) -- (5);
    \draw[thick] (37) -- (4);
    \draw[thick] (37) -- (3);

   \draw[thick] (55) -- (30);
   \draw[thick] (55) -- (37);
   \draw[thick] (55) -- (7);
   \draw[thick] (55) -- (1);

\end{tikzpicture}
\caption{The frame $\mathbf{C}_{u,n}$}
\label{fig:china}
\end{figure}
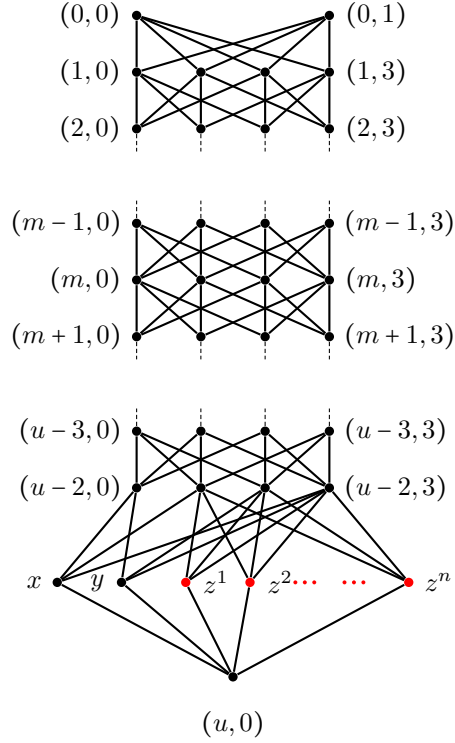

\begin{lemma}
	$\mathbf{C}_{2^{s-1}+1,2^s}$ is a $\mathsf{Cheq}$-frame for each $s\geq 2$.
\label{lem:cc}	
\end{lemma}

\begin{proof}
For every $s\geq 2$, we aim to prove that $\mathbf{C}_{2^{s-1}+1,2^s}$ is a \emph{p}-morphic image of $\mathbf{V}_{2^{s-1}+1}$.
	
For any $a\in \mathbf{V}_{2^{s-1}+1}$, let $\#_0(a)$ be the number of $w_0$'s in $a$. If $\#_0(a) \geq 1$, let $\#_1(a)$ be the number of $w_1$'s to the left of the leftmost $w_0$ in $a$, $\#_2(a)$ be the number of $w_1$'s to the right of the rightmost $w_0$ in $a$. 
If $a\in O:=\{a: \#_0(a)=2^{s-1}\} \setminus \{(w_1, w_0, \ldots, w_0), (w_0,\ldots, w_0, w_1)\}$, then choose a one-to-one map $o(a):  O\to \{1,\ldots, 2^s\}$.
Define a map $f: \mathbf{V}_{2^{s-1}+1} \to  \mathbf{C}_{2^{s-1}+1,2^s}$ by
$$f(a)=\begin{cases}
	(0,0), &\text{if } \#_0(a)=0 \text{ and the number of } w_1 \text{ is even},\\	
	(0,1), &\text{if } \#_0(a)=0 \text{ and the number of } w_1 \text{ is odd},\\		
	(\#_0(a),0), &\text{if } 1\leq\#_0(a)\leq 2^{s-1}-1 \text{ and } 2\nmid \#_1(a) \text{ , }2\nmid \#_2(a),\\	
	(\#_0(a),1), &\text{if } 1\leq\#_0(a)\leq 2^{s-1}-1 \text{ and } 2\nmid \#_1(a)
	    \text{ , }2\mid \#_2(a),\\
	(\#_0(a),2), &\text{if } 1\leq\#_0(a)\leq 2^{s-1}-1 \text{ and } 2\mid \#_1(a)\text{ , }2\nmid \#_2(a),\\
	(\#_0(a),3), &\text{if } 1\leq\#_0(a)\leq 2^{s-1}-1 \text{ and } 2\mid \#_1(a)\text{ , }2\mid \#_2(a),\\
	
	x, &\text{if } x=(w_1, w_0, w_0, \ldots, w_0),\\	
	y, &\text{if } x=(w_0, w_0, \ldots, w_0, w_1),\\	
	z^i, &\text{if } o(a)=i,\\
	(2^{s-1}+1,0), &\text{if } \#_0(a)=2^{s-1}+1.
\end{cases}$$ 
It is clear that $f$ is onto. Assume that $a < b$ with $ \#_0(a)=\#_0(b)+1$ in $\mathbf{V}_{2^{s-1}+1}$.
If $\#_0(a)=1$, then $\#_0(b)=0$ and thus $f(b)=(0,0)$ or $(0,1)$, in each case we have that $f(a)<f(b)$ .
If $2\leq \#_0(a) \leq 2^{s-1}$, then it is impossible to change both $\#_1$ and $\#_2$ from $a$ to $b$, so $\#_i(a)=\#_i(b)$ for at least one $i\in \{1, 2\}$ and $f(a) < f(b)$. 
If $\#_0(a)=2^{s-1}+1$, then $f(a)$ is the root.
In general case of $a<b$ in $\mathbf{V}_{2^{s-1}+1}$, it is easy to have a chain $a^0<a^1<\ldots <a^p$ such that $a^0=a$, $a^p=b$, and $\#_0(a^q)=\#_0(a^{q+1})+1$ for $0\leq q \leq p-1$, hence $f(a)=f(a^0)<f(a^1)<\ldots <f(a^p)=f(b)$.

Assume that $f(a)< t$ with $d\bigl(f(a)\bigr)=d(t)+1$ in $\mathbf{C}_{2^{s-1}+1,2^s}$. 
If $f(a)$ is the root of $\mathbf{C}_{2^{s-1}+1,2^s}$, then $t\in \{x,y,z^1,z^2,\ldots,z^n\}$, thus $b=f^{-1}(t) >a$ and $f(b)=t$.

If $f(a)=(m,0)$ and $t=(m-1,0)$ for $2\leq m \leq 2^{s-1}$, let \(n_1\) and \(n_2\) be the first and second coordinates of \(a\) equal to \(w_0\), respectively, from left to right. Define $b=(b_0,b_1,\ldots,b_{2^{s-1}+1})$ as follows:
$$b_n=\begin{cases}
	w_1, &\text{if } n=n_1 \text{ and the number of } w_1 \text{ to the left of the }a_{n_2} \text{ in } a \text{ is even},\\	
   w_2, &\text{if } n=n_1 \text{ and the number of } w_1 \text{ to the left of the }a_{n_2} \text{ in } a \text{ is odd},\\
   a_n, &\text{otherwise}.
\end{cases}$$ 
Then $a < b$ and $f(b)=(m-1, 0)=t$. By the same argument, the remaining cases for $2\leq m \leq 2^{s-1}-1$ follow analogously.

If $f(a)=(1,0)$ and $t=(0,0)$, then $\#_1(a)+\#_2(a)$ is even. Define $b=(b_0,b_1,\ldots,b_{2^{s-1}+1})$ as follows:
$$b_n=\begin{cases}
	w_2, &\text{if } a_n=0,\\	
    a_n, &\text{otherwise}.
\end{cases}$$ 
Then $a < b$ and $f(b)=(1, 0)=t$. By the same argument, the remaining cases following immediately. 

In general case of $f(a)<t$ in $\mathbf{C}_{2^{s-1}+1,2^s}$, it is easy to have a chain $t^0<t^1<\ldots <t^p$ such that $t^0=f(a)$, $t^p=t$, and $d(t^q)=d(t^{q+1})+1$, hence we can find $b$ such that $a < b$ and $f(b)=t$. 	

As a conclusion, $f$ is a \emph{p}-morphism from $\mathbf{V}_{2^{s-1}+1}$ to $\mathbf{C}_{2^{s-1}+1,2^s}$.
\end{proof}

The following lemma is due to \cite[Claim~16]{Fon07}.

\begin{lemma}[Fontaine]
Let $\mathbf{F}$ be a finite frame and it is a p-morphic image of some  $\mathbf{V}_n$. Suppose that every point in $\mathbf{F}$ of depth two has branching degree two and that every point in $\mathbf{F}$ of depth greater than two has branching degree greater than or equal to three. Then for every $x\in \mathbf{F}$, we have that $b(x) \leq 2 \times 3^{d(x)-1}$.
\label{Fon}
\end{lemma}

\begin{corollary}
$\mathbf{C}_{s, 2 \times 3^s}$ is not a $\mathsf{Cheq}$-frame.	
\label{notcheq}
\end{corollary}
\begin{proof}
Because $\{V_n: n\geq 1\}$ is a class of frames closed under rooted
generated subframes, then for every finite rooted frame, it is a $\mathsf{Cheq}$-frame iff  it is a \emph{p}-morphic image of some $\mathbf{V}_n$ by Jankov-de Jongh Theorem. The root of $\mathbf{C}_{s, 2 \times 3^s}$ has depth of $s+1$, its branching degree is $2 \times 3^s+2$, so $\mathbf{C}_{s, 2 \times 3^s}$ is not a \emph{p}-morphic image of any $\mathbf{V}_n$ by Lemma~\ref{Fon} and so is not a $\mathsf{Cheq}$-frame.
\end{proof}

\begin{theorem}
$\mathsf{Cheq}$ is not finitely axiomatizable.	
\label{thm:cheq}
\end{theorem}

\begin{proof}
Suppose $\varphi(p_1, p_2, \ldots, p_s)$ is a formula with $s$ variables axiomatizing $\mathsf{Cheq}$  for $s\geq 2$. Let $s_1= 2^{s-1}+1$, then $\mathbf{C}_{s_1,2 \times 3^{s_1}}$ is not a $\mathsf{Cheq}$ -frame.
Let $V_0$ be the valuation such that $(\mathbf{C}_{s_1,2 \times 3^{s_1}}, V_0)$  refutes $\varphi$. 
By Dirichlet's drawer principle, there exist $1\leq n<m \leq 2 \times 3^{s_1}$, such that $z^n$ and $z^m$ agree on all propositional variables $p_i$ for $1 \leq i \leq s$. 
Without less of generality, we may assume that $n=2 \times 3^{s_1}-1$ and $m=2 \times 3^{s_1}$. 
Define a valuation $V_1$ on $\mathbf{C}_{s_1,2 \times 3^{s_1}-1}$ by $V_1(p_i) = V_0(p_i)\setminus \{z^{2 \times 3^{s_1}}\}$. 
It is easy to construct a \emph{p}-morphism that identifies $z^{2 \times 3^{s_1}-1}$ and $z^{2 \times 3^{s_1}}$, this map is a \emph{p}-morphism from $(\mathbf{C}_{s_1,2 \times 3^{s_1}},V_0)$ to $(\mathbf{C}_{s_1,2 \times 3^{s_1}-1},V_1)$ and so $\mathbf{C}_{s_1,2 \times 3^{s_1}-1} \nvDash \varphi$. We can keep reducing the number of $z^j$ till only $2^s$ remain. As a conclusion, $\mathbf{C}_{s_1,2^s} \nvDash \varphi$, but $\mathbf{C}_{2^{s-1}+1,2^s}$ is a $\mathsf{Cheq}$-frame by Lemma~\ref{lem:cc}. Contradiction!

Therefore, $\mathsf{Cheq}$ is not axiomatizable with $s$ variables for any $s\geq 2$ and so $\mathsf{Cheq}$ is not finitely axiomatizable.
\end{proof}

\begin{corollary}
	$\mathsf{L}_{\infty}$ is not finitely axiomatizable.
\end{corollary}
\begin{proof}
By Theorem~\ref{thm:cheq}, $\mathsf{L}_{\infty}=\sigma(\mathsf{Cheq})$ is not finitely axiomatizable.
\end{proof}

\section{Decidability of $\mathsf{Cheq}$}
The non-finite axiomatizability result now has a different proof via the undecidability. If $\mathsf{Cheq}$ is undecidable, then $\mathsf{Cheq}$ is not finitely axiomatizable since it has finite model property. We provide a reduction from the decidability problem for $\mathbf{Med}$ to $\mathbf{Cheq}$.  

\subsection{Reduction}
For any formula $\varphi$, let $p, q$ be two distinct propositional variables that are not occurred in $\varphi$. Define $\psi_0(p)$ and $\psi_1(p,q)$ as follows:
\[\psi_0 (p):= p \vee \neg p \quad \psi_1 (p,q):=\neg p \vee \neg q.	
\]
Then we define a formula $C(\varphi): = (\varphi^{\psi_0}\to \psi_1)\to  \psi_1$.

The formula $\psi_0$ and $\psi_1$ admit intuitive geometric interpretations. Following our previous discussion, we regard $\mathbf{V}_n$ as a $n$-hypercube, its subface also generates a sub-hypercube. If $\psi_1$ fails on some hypercube, this means that there are two vertices $v_0$ and $v_1$ witnessing $p$ and $q$, respectively.  They need not be a diagonal of the whole hypercube, but $v_0$ and $v_1$ will generate, by serving as a diagonal, a subhypercube. In this subhypercube, only the vertex $v_0$ witness $p$. In a hypercube, when only the vertex $v$ witness $p$, the faces on which $\psi_0$ fails are precisely those faces containing $v$, they form a simplex and the Medvedev frame appears here.

Our construction of the countermodel consists in carving out a simplex in a subhypercube inside a hypercube, or in using a simplex to reconstruct a hypercube. The role of Friedman translation here is to make $\varphi$ and $\varphi^{\psi_0}$ correspond to two different geometric objects. Since $\mathsf{Med}$ corresponds to a simplex, which is in fact the simplest geometric object, the logics with a geometric meaning, by this I mean the logics that can be characterized by geometric objects, admit what seems to be a feasible procedure of carving out a simplex or other geometric objects. This suggests that many logics can be given translations by the same method, thereby construct a reduction from their undecidability to that of logics of known geometric objects, such as $\mathsf{Med}$ and $\mathsf{Cheq}$.

\begin{theorem}
For any formula $\varphi$, $\varphi \in \mathsf{Med}$ iff $C(\varphi)\in \mathsf{Cheq}$.	
\label{thm:cheqmed}
\end{theorem}

\begin{proof}
($\Rightarrow$) If $C(\varphi) \notin \mathsf{Cheq}$, then there is a model $\mathcal{M}$ of $\mathbf{V}_n$ and $x\in \mathbf{V}_n$ such that $x\vDash_{\mathcal{M}} \varphi^{\psi_0} \to  \psi_1$ but $x\nvDash_{\mathcal{M}} \psi_1$. There exists $y\geq x$ such that, $y \nvDash_{\mathcal{M}} \psi_1$ and no $y'> y$ with  $y' \nvDash_{\mathcal{M}}\psi_1$. 
Because $x\vDash_{\mathcal{M}} \varphi^{\psi_0} \to  \psi_1$ and $x\leq y$, we obtain that $y\vDash_{\mathcal{M}} \varphi^{\psi_0} \to  \psi_1$. If $y \vDash_{\mathcal{M}} \varphi^{\psi_0}$, then $y\vDash_{\mathcal{M}} \psi_1$, a contradiction. So $y\nvDash_{\mathcal{M}} \varphi^{\psi_0}$. 
Let $\mathbf{P}_0:=(\{x\in y\uparrow : x\nvDash_{\mathcal{M}} \psi_0\},\leq)$. If $y$ is maximal in $\mathbf{V}_n$, then $y\vDash_{\mathcal{M}} p$ since $y\nvDash_{\mathcal{M}} \psi_1$. Hence $y\vDash_{\mathcal{M}} \psi_0$ and so $y\vDash_{\mathcal{M}} \varphi^{\psi_0}$ because of the definition of Friedman translation, a contradiction. Then $y$ is not a maximal point and so $\mathbf{P}_0 \neq \emptyset$. By Pruning Lemma~\ref{lem:pruning}, $y\vDash_{\mathcal{M}} \varphi^{\psi_0}$ implies $y \nvDash_{\mathcal{M}_0} \varphi$ in the submodel ${\mathcal{M}_0}$ restrict on $\mathbf{P}_0$. Now consider the frame $\mathbf{P}_0$, its root $y\nvDash_{\mathcal{M}_0} \psi_1$. Because $\psi_1=\neg p \vee \neg q$, then there are maximal points $x'$ and $x''$ above $y$ in $\mathbf{V}_n$, such that $x' \vDash_{\mathcal{M}_0} p$ and $x'' \vDash_{\mathcal{M}_0} q$. Define $y'=(y'_1, \ldots, y'_n)$ as follows: 
$$y'_i=\begin{cases}
	x'_i, &\text{if } x'_i=x''_i,\\	
    w_0, &\text{otherwise}.
\end{cases}$$     
Then $y'_i=w_0$ implies that $\{x'_i, x''_i\}=\{w_1,w_2\}$, and so $y_i=w_0$. If $y'_i\neq w_0$, then $y'_i=x'_i\geq y_i$. Thus $y\leq y'$. Since $y'$ can see both $x'$ and $x''$, then $y'\nvDash_{\mathcal{M}} \psi_1$. But no $y' > y$ with $y' \nvDash_{\mathcal{M}} \psi_1$, so $y'=y$. If there is another maximal point $z\neq x'$ above $y$ with $z\vDash_{\mathcal{M}_0} p$, then we can define $y''$ as follows:
$$y''_i=\begin{cases}
	x''_i, &\text{if } x''_i=z_i,\\	
    w_0, &\text{otherwise}.
\end{cases}$$  
By the same argument, we obtain that $y''=y=y'$ and so $z=x'$, a contradiction. Thus only one maximal point $x' \vDash_{\mathcal{M}_0} p$ in $y\uparrow$. So $\mathbf{P}_0=\{x: y\leq x < x'\}$. For any $x\in \mathbf{P}_0$, if $y_i=x'_i$, then $x_i=x'_i$. If $y_i \neq x'_i$, then $x_i =w_0$ or $x_i= x'_i$. Thus $\mathbf{P}_0$ is isomorphism to some Medvedev frame.
  Therefore, $\varphi \notin \mathsf{Med}$ since $y\nvDash_{\mathcal{M}_0} \varphi$ on a model of $\mathbf{P}_0$.

($\Leftarrow$) If $\varphi \notin \mathsf{Med}$, then there is a model $\mathcal{M}$ of Medvedev frame $\mathbf{P}_0(n)$ such that $[n] \nvDash_{\mathcal{M}} \varphi$. 
In the frame $\mathbf{V}_n$, define a model $\mathcal{M}_1:=(\mathbf{V}_n, V)$, where $V$ is a valuation such that $V(p)=\{(w_1,\ldots, w_1)\}$ and $V(q)=\{(w_2,\ldots, w_2)\}$. 
Let $\mathbf{P}_1:=\{x\in \mathbf{V}_n : x  \nvDash_{\mathcal{M}_1} \psi_0\}$, then $\mathbf{P}_1=(w_1,\ldots,w_1)\downarrow \setminus (w_1,\ldots,w_1)$ in $\mathbf{V}_n$, which is obtained from the downset of $(w_1,\ldots,w_1)$ by removing the top. It is obvious that $\mathbf{P}_1 \cong \mathbf{P}_0(n)$, and so $(w_0,\ldots, w_0) \nvDash_{\mathcal{M}} \varphi$. 
By Pruning Lemma, $(w_0,\ldots, w_0) \nvDash_{\mathcal{M}} \varphi$ implies $(w_0,\ldots, w_0) \nvDash_{\mathcal{M}'_1} \varphi^{\psi_0}$, 
where $\mathcal{M}'_1 $ is a model on $\mathbf{V}_n$ and agrees with $\mathcal{M}$ and $\mathcal{M}_1$ on their respective valuation. 
For any $x\in \mathbf{V}_n$,  $x \nvDash_{\mathcal{M}'_1} \psi_1$ iff $x= (w_0,\ldots, w_0)$. So $(w_0,\ldots, w_0) \vDash_{\mathcal{M}'_1} \varphi^{\psi_0}\to  \psi_1$. Together with $(w_0,\ldots, w_0) \nvDash_{\mathcal{M}'_1} \psi_1$, we obtain that $ (w_0,\ldots, w_0) \nvDash_{\mathcal{M}'_1} (\varphi^{\psi_0}\to \psi_1)\to \psi_1$. As a conclusion, $C(\varphi)\notin \mathsf{Cheq}$.	
\end{proof}

\begin{corollary}
	If $\mathsf{Med}$ is undecidable, then $\mathsf{Cheq}$ is undecidable either and, consequently, is not finitely axiomatizable.
\end{corollary}

\begin{proof}
By the Theorem~\ref{thm:cheqmed}, the decidability of $\mathsf{Cheq}$ implies the decidability of $\mathsf{Med}$. Since $\mathsf{Cheq}$ has the finite model property, then the undecidability of it implies the non-axiomatizability. 
\end{proof}

Recently, since the undecidability of $\mathsf{Med}$ has been proved by AI, see \cite{Ak26, Paw26}, then
\begin{theorem}
	$\mathsf{Cheq}$ is undecidable.
\label{cheqund}
\end{theorem}

\bigskip
\noindent\textbf{Acknowledgements.}
I thank Benedikt L\"owe for his support of my work on $\mathsf{Cheq}$ and for his expert guidance, Nick Bezhanishvili for invaluable advice, and Ga\"elle Fontaine for helpful suggestions. I would also thank Johan van Benthem, one of those who posed the open problem considered here, for his continued interest and discussions.

\newpage
\nocite{*} 

\bibliography{wpref.bib}

\end{document}